\documentclass[sn-mathphys-num]{sn-jnl}
\usepackage{tikz-cd}
\usepackage{braket}
\usepackage{tikz}
\usetikzlibrary{decorations.markings}
\usetikzlibrary{arrows.meta}
\usepackage{todonotes}
\usepackage{graphicx}%
\usepackage{multirow}%
\usepackage{amsthm,amssymb} % Add this package
\usepackage{amsmath,amssymb,amsfonts}%
\usepackage[capitalize,nameinlink]{cleveref}
\usepackage{amsthm}%
\usepackage{mathrsfs}%
\usepackage[title]{appendix}
\usepackage[dvipsnames,svgnames]{xcolor}%
\usepackage{mathtools}

\usepackage{dsfont}

\newcommand{\op}[1]{\operatorname{#1}}
\newcommand{\din}{\delta_{\op{in}}}
\newcommand{\dout}{\delta_{\op{out}}}
\newcommand{\dfree}{\delta_{\op{free}}}

\newcommand{\ftwo}{{\mathbb{F}_2}}

\newcommand{\C}{\mathbb{C}}
\newcommand{\im}{\op{im}}

\newcommand{\gate}{\underline{\op{C}}}
\newcommand{\gateZ}{\underline{\op{C}}}

\newcommand{\Zn}[1]{\mathbb{Z}/#1}

\newtheorem{theorem}{Theorem}[section]
\newtheorem{lemma}{Lemma}[section]
\newtheorem{definition}{Definition}[section]

\newtheorem{proposition}{Proposition}[section]

\theoremstyle{remark} % Define remark style
\newtheorem{remark}{Remark}[section]
\newtheorem{example}{Example}[section]

\newcommand{\mathcolorbox}[1]{\colorbox{yellow}{$\displaystyle #1$}}
\usepackage{soul,color}
\soulregister\cite7
\soulregister\ref7
\soulregister\pageref7
 \renewcommand{\mathcolorbox}[1]{#1}
\renewcommand{\hl}[1]{#1}

\begin{document}
\title{Cups and Gates I: Cohomology invariants and logical~quantum operations}
 \author[1]{\fnm{Nikolas P.} \sur{Breuckmann}}

 \author[2,3,4]{\fnm{Margarita} \sur{Davydova}}

\author[5]{\fnm{Jens N.} \sur{Eberhardt}}

\author[2]{\fnm{Nathanan} \sur{Tantivasadakarn}}

\affil[1]{\small \orgdiv{School of Mathematics}, \orgname{University of Bristol}, \orgaddress{\street{Fry Building, Woodland Road}, \city{Bristol}, \postcode{BS8~1UG}, \country{United Kingdom}}}

 \affil[2]{\small \orgdiv{ Walter Burke Institute for Theoretical Physics and Department of Physics}, \orgname{California Institute
of Technology}, \orgaddress{\small \street{1200 E California Blvd}, \city{Pasadena}, \postcode{91125}, \state{CA}, \country{USA}}}

 \affil[3]{\small \orgdiv{Institute for Quantum Information and Matter}, \orgname{California Institute
of Technology}, \orgaddress{\street{1200 E California Blvd}, \city{Pasadena}, \postcode{91125}, \state{CA}, \country{USA}}}

\affil[4]{\small Department of Physics, Massachusetts Institute of Technology, Cambridge, 02139, MA, USA}

 \affil[5]{\small \orgdiv{Institut für Mathematik}, \orgname{Johannes Gutenberg University of Mainz}, \orgaddress{\street{Staudingerweg 9}, \city{Mainz}, \postcode{55128}, \country{Germany}}}

\abstract{
We take initial steps towards a general framework for constructing logical gates in general quantum CSS codes.
Viewing CSS codes as cochain complexes, we observe that \emph{cohomology invariants} naturally give rise to diagonal logical gates.
We show that such invariants exist if the quantum code has a structure that relaxes certain properties of a differential graded algebra.
We show how to equip quantum codes with such a structure by defining \emph{cup products} on CSS codes. 
The logical gates obtained from this approach can be implemented by a constant-depth unitary circuit.
In particular, we construct a $\Lambda$-fold cup product that can produce a logical operator in the $\Lambda$-th level of the Clifford hierarchy on $\Lambda$ copies of the same quantum code, which we call the \emph{copy-cup gate}.
For any desired $\Lambda$, we can construct several families of quantum codes that support gates in the $\Lambda$-th level with various asymptotic code parameters.

}

\maketitle

\tableofcontents

\section{Introduction}

The existence of a quantum memory does not fully solve the problem of fault-tolerant quantum computation. 
It is equally essential to develop efficient ways for realizing a universal set of logical gates in a fault-tolerant manner.

With the advent of new classes of quantum low-density parity check (qLDPC) codes with superior parameters, finding fault-tolerant logical gates in these codes is an important and open challenge, and most methods currently used to construct such gates have been rather ad-hoc.

Logical gates that can be implemented by constant-depth local circuits are naturally fault-tolerant because the error spread is bounded. In fact, they present arguably the most efficient way to do fault-tolerant gates. 
A prime example of codes that have such gates are the color codes, which have been proposed on triangulations on manifolds with a particular structure~\cite{BombínCC,BombínGaugeCC,BeverlandKubicaSimplified,VasmerBrown3D}. 
For codes defined on manifolds, the level in the Clifford hierarchy for gates implementable by a constant-depth geometrically-local circuits is known to be limited by the dimensionality of the manifold via the Bravyi-K\"onig theorem~\cite{Bravyi_2013}.\footnote{While there has been some progress for more general qLDPC codes~\cite{burton2020limitationstransversalgateshypergraph}, an analogous version of the Bravyi-K\"onig theorem is not known.} 
A number of techniques have been put forward for fault tolerant non-constant depth unitary gates, most of which yield higher spacetime overhead. 
This includes code switching~\cite{Anderson_2014,Bombin2016Apr}, lattice surgery~\cite{HorsmanSurgery,CohenSurgery,CowtanSurgery}, pieceable circuits~\cite{YoderPartitioning} magic state injection~\cite{knill2004faulttolerantpostselectedquantumcomputation}, as well as numerical/computational search~\cite{WebsterSearch}.

For high-rate and high-distance qLDPC codes~\cite{qldpc_review}, the issue of finding low-overhead fault-tolerant gate implementations is even more challenging, as very few of the existing techniques can be readily applied. Some progress has been made in this direction recently, including non-manifold generalizations of color codes~\cite{PinCodes,RainbowCodes}, using automorphisms of the code to realize certain classes of constant-depth unitary gates \cite{Breuckmann_2024,Quintavalle_2023}, and adaptations of lattice surgery methods~\cite{CohenSurgery,CowtanSurgery}.
All these methods have drawbacks, such as only applying to specific codes and in restricted settings, giving restricted gate sets, requiring large spacetime overhead, or admitting case-by-case fault-tolerance proofs.

In this paper, we make the first steps toward a general framework for constructing diagonal logical gates in CSS quantum codes.
If the quantum code is LDPC, our construction can be used to design gates implementable using constant-depth local circuits.
Our framework leverages the close connection between CSS quantum codes and homological algebra.
In particular, \emph{cohomology invariants} naturally give rise to diagonal logical gates.
This can be seen as follows.

In the homological picture, a CSS code corresponds to a chain (cochain) complex. The code states correspond to the \emph{cohomology classes} $[\gamma_i]$, $i = 1,..,k$. Colloquially, these can be understood as classes of nontrivial logical operators up to deformations by stabilizers. 
Here, $\gamma_i$ is a representative \emph{cocycle} whose non-zero entries denote the locations of the Pauli-$X$ operators that together form the corresponding logical $X$ operator. We will denote the result of applying the all zero state $\ket{0}^{\otimes N}$ with the logical X-operator as $\ket{\gamma_i}$. 
The code states formed by starting with $\ket{\gamma_i}$ and further symmetrizing over all possible $X$-stabilizers (equivalently, the equal-weight superposition of $\ket{\gamma_i + b}$ over all possible coboundaries $b$) is labeled as  $\ket{[\gamma_i]}$. 

We then consider an operation $\psi$ mapping to $\mathbb{R} / \mathbb{Z} $ with a property
\begin{equation}
    \psi(\gamma') = \psi(\gamma) \equiv \psi([\gamma_i])  \ \  \text{ if } \gamma, \gamma' \in [\gamma].
\end{equation}
Operations with this property are cohomology invariants, i.e.\ their action depends only on the cohomology class. Because of the one-to-one correspondence between the cohomology classes and the codestates, such an action corresponds to a diagonal logical gate, as it applies a different phase to each logical state.  

We will restrict ourselves to only such operations $\psi$ that, on top of giving a cohomology invariant, are also defined on arbitrary states and not only on ones corresponding to cocycles.  Such operations can be realized by a circuit specified by its action on a state $\ket{c}$ in the computational basis:
\begin{equation}
    U_\psi \ket{c} = e^{2 \pi i \psi(c)}\ket{c}.
\end{equation}
Because each codestate $\ket{[\gamma_i]}$ is a superposition of states labeled by cocycles $\ket{\gamma_i'}$ with $\gamma_i' \in [\gamma_i]$ and $\psi(\gamma_i') = \psi([\gamma_i])$, the unitary $U_\psi$ acts with a well-defined phase $e^{2 \pi i \psi([\gamma_i])}$ on all such states in the superposition. Thus, it indeed preserves the codestate and has a diagonal logical action. 
In this paper, we show how to construct such operations at the logical level under certain conditions on the code and how to find appropriate constant-depth circuits to implement them.

Homological algebra and cohomology invariants in particular, are a very well-developed and rich subject in mathematics that we wish to use for constructing such operations.
We want to emphasize that our approach evades the issue of being ad hoc, as the cohomological view is general, structured, and classifiable.

The cohomological approach has another advantage.
At a high level, cohomology provides a connection between global features and local connectivity.
As we will make precise later, in our setting, `global features' correspond to the logical level, whereas `local connectivity' corresponds to the connectivity of circuits and checks.
Cohomology hence provides us with a recipe that allows us to extract the physical circuit from the logical representation.
However, all calculations can in principle be done at the logical level only, making them much simpler.
The physical circuits that we obtain can become quite non-intuitive and it is doubtful that they could have been found through human intuition or brute-force computer search, thus highlighting the utility of a more structured perspective over brute-force methods.

Our main tool for constructing cohomology operations $\psi$ is the \emph{cup product} $\cup$ and a function $\int$ called the \emph{integral} that assigns to a cohomology class a number in $\Zn{m}$.\footnote{Readers familiar with calculus on manifolds and de Rham cohomology may think of the cup product as a discrete analog of the wedge product, which we can use to construct `forms' over which we integrate, giving us topological invariants.}
We take this geometric idea and propose a generalization of it in order to define various cohomology invariants on codes.\footnote{We may glibly point out that geometric ideas have a long history when it comes to defining quantum gates: evaluating the Jones polynomial and Turaev--Viro invariants are deeply connected to gate implementations in topological codes and fault-tolerant gates have a natural interpretation in terms of fiber bundles \cite{gottesman2013fibre}.}
In particular, we show that there exist well-defined notions of cup products and integrals for classical codes, as well as the corresponding quantum codes obtained via both tensor products (also called `hypergraph products' in the context of quantum codes~\cite{Tillich_2014}) and balanced products of classical codes, if certain assumptions are fulfilled.
We note that balanced product codes include several interesting code families, such as all bivariate bicycle codes and good qLDPC codes.

In this paper we give a first example of a construction of families of codes with well-controlled cohomology invariants.
This construction involves some number, denoted by an integer $\Lambda$, of copies of the same quantum code. 
We formulate the conditions of the existence of a $\Lambda$-fold cup product in a general CSS quantum code that results in diagonal logical gates that we call \emph{copy-cup gates}. We show that these gates can lie at the $\Lambda$-th level of the Clifford hierarchy.

\subsection{Overview of the main results}

A CSS quantum code $\mathcal C$ can be identified with a three-term cochain complex $C(\Zn{2})$
\[\begin{tikzcd}
	{C^{0}} & {C^1} & {C^{2}} 
	\arrow[from=1-1, to=1-2]
	\arrow["{\delta^{0}}", from=1-1, to=1-2]
	\arrow["{\delta^1}", from=1-2, to=1-3]
\end{tikzcd}\]
where we place physical qubits in correspondence with the basis of $C^1$, and $X$($Z$) checks in correspondence with the bases of $C^0$ ($C^2$). 
The linear maps $\delta^i$ are the coboundary operations setting relations between qubits and checks. 
The logical states of the code are labeled by the elements of the cohomology class $[\gamma] \in H^1$.

Consider a multilinear operation $\Psi$ that is defined on cochains belonging to $\Lambda$ copies of the same quantum code $\mathcal C$:
$$\Psi: C^{1}(\Zn{2})\times \dots \times C^{1}(\Zn{2})\to \Zn{m}.$$
We are interested in operations that also correspond to cohomology invariants, i.e. which induce the action
$$\Psi: H^{1}(\Zn{2})\times \dots \times H^{1}(\Zn{2})\to \Zn{m}.$$
Then, we may apply $\Psi$ to $\Lambda$ copies of the same cell complex to construct a transversal gate. Namely, $\Psi(\gamma'_1 ,..., \gamma'_\Lambda) = \Psi (\gamma_1, ... \gamma_\Lambda)$ for $\gamma'_i \in [\gamma_i]$.  As we discussed above, this allows us to put such an operation in correspondence with a diagonal logical action. 

In the paper, we lay some of the general groundwork for this approach. However, the most natural example of a multilinear operation on $\Lambda$ arguments that is a cohomology invariant that is also defined at a cochain level is a cup product. For this operation, we assume $m = 2$ and  formulate the following explicit result: 

\setcounter{section}{6}
\begin{theorem}[Informal; copy-cup gates]
Suppose a CSS quantum code $\mathcal C$ has a $\Lambda$-fold cup product operation:
$$\Psi_{\op{C} \cup}(c_1,\dots,c_\Lambda)= \int c_1\cup \dots \cup c_\Lambda$$
whose action is nontrivial. 
Then the following statements hold:
\begin{enumerate}
    \item The unitary
    $$U_{\op{C}\!\cup} = \sum_{\{c_i\}}(-1)^{\int\!c_1\cup\dots\cup c_\Lambda}\ket{c_1,\dots,c_d}\bra{c_1,\dots,c_\Lambda}$$
    preserves the codespace and has diagonal logical action at $\Lambda$-th level  of the Clifford hierarchy.
    \item The above unitary can be realized by the constant-depth circuit  $$\gateZ_{\op{C}\!\cup}=\prod_{\{x_i\}}(\op{C}^{\Lambda-1}\!Z_{x_1,\dots,x_\Lambda})^{\int x_1\cup \dots \cup x_\Lambda}.$$ 
    which involves the gates at $\Lambda$-th level of the Clifford hierarchy.
\end{enumerate}
\end{theorem}
\setcounter{section}{1}

We also discuss conditions under which a CSS quantum code can have such an operation. 
The cup product for quantum codes obtained from tensor and balanced products of \emph{classical} codes can be induced by an associated cup product on the constituent classical codes. 
We derive the conditions for the cup product structure on the classical codes so that an associated cup product on the quantum code has desirable properties. 
In particular, a sufficient condition can be stated as follows: if the cup product in the classical code obeys the Leibniz rule up to integration, the integrated cup product in the quantum code is a cohomology invariant. 
We show that this condition can be conveniently summarized as a set of conditions on the parity check matrix of the code that bear some similarities to the triorthogonality condition~\cite{Bravyi_2012}.

Using these ideas, one can construct a number of code families where the copies of the same code admit transversal gates at $\Lambda$-th level of the Clifford hierarchy.  Moreover, if one finds codes with torsion homology groups, such an operation can give nontrivial logical action on a single code. 
We provide several examples of code families that have cup product operations. 
The simplest ones are on manifolds and include families of $\Lambda$-dimensional toric codes (with parameters $[[n, \Lambda^2, \Theta(\sqrt{n})]]_{n}$) and hyperbolic codes with \hl{polynomial rate} (for $\Lambda$ even, we consider $\biguplus_\Lambda H_g^{\Lambda/2}$ giving parameters \hl{$[[n = \Theta(\Lambda^2 g), \Theta(n^{\frac {2} {\Lambda}}), \Theta(\log^\alpha n)]]_{g}$}). 
In these families, $\Lambda$ copies of a respective code will have logical gates at $\Lambda$-th level of the Clifford hierarchy via the cup product operation. 
We also show that certain bivariate bicycle codes can satisfy this condition for a $\Lambda = 2$ cup product and conjecture that it is possible to have products of group algebra codes that have $\Lambda \geq 3$ cup products. Finally, if certain assumptions can be fulfilled, hypergraph and balanced products of Sipser-Spielman codes with constant rate and polynomial distance can also have $\Lambda$-fold cup products.  The condition that we formulate on the existence of a cup product cohomology operation in quantum codes is a convenient starting point for searching for more examples of qLDPC codes with corresponding gates.

\emph{Note added:}  we recently became aware of upcoming work \cite{lin2024transversal} where a general framework involving cup products is used to construct transversal non-Clifford gates on sheaf qLDPC codes; this formalism was also used in upcoming work \cite{golowich2024quantum}.

\subsection{Connection to earlier literature}

The use of cup products and other standard cohomological operations in constructing transversal codes has been explored to some degree in the context of topological codes. For example, explicit formulas are known to implement transversal $\op{CZ}$ gate in two copies of the 2D toric code \cite{Yoshida16,Potter17,Barkeshli23} and transversal $\op{CCZ}$ gate in three copies of the 3D toric code \cite{Chen2023higher,Barkeshli23,Wang2023} and the cup product was used to understand the logical action of the transversal $T$ gate in the 3D color code~\cite{Zhu2023}. \hl{The connection between the multiplicity of the cup product and the Clifford hierarchy has been pointed out in Refs.~\cite{Barkeshli23,Barkeshli24_1,Barkeshli24_2}. Ref.~\cite{Zhu2023} first constructed non-Clifford gates on high-rate qLDPC codes that are associated with cup products on hyperbolic manifolds  (this result was subsequently improved in \cite{RainbowCodes}).} The addressable transversal $\op{CZ}$ gate in multiple copies of the 3D toric code has also been explored using the slant product in \cite{Yoshida17,Barkeshli23}, where the gate can also be interpreted in terms of sweeping a codimension-2 symmetry defect.
Gates in the 3D fermionic toric code have similarly been explored in this context, \cite{WebsterBartlett18,Kobayashi23,Barkeshli24_2,Fidkowski24}, albeit containing non-explicit formulas for certain gates due to the pumping of chiral invertible phases.
In addition, in \cite{Chen23}, the Pontryagin square was used to realize a  logical $\op{S}$ gate in the 4D loop-only toric code using physical $\op{CS}$ gates. 

To the best of our knowledge, we are not aware of progress in adapting these cohomological operations approach to more general quantum codes beyond manifolds, i.e. to homological codes (apart from forthcoming papers \cite{lin2024transversal,golowich2024quantum}).

\subsection{Outlook}

\noindent \textbf{(1) Construction of more general cohomology invariants for quantum gates and their classification.}
While this paper has made some advances towards a general homological framework for constant-depth circuit implementable quantum gates, the space of cohomology operations has not been explored to a large extent in this context. Examples of other operations that could be adapted in search for other cohomology invariants in the context of quantum gates include Bockstein homomorphism, Steenrod square, and
Pontryagin square. 
In addition, there could also be unstable cohomology operations, whose full classification does not yet exist. Finally, one would also need to find conditions on quantum codes that permit such operations and furthermore, lead to interesting gates.

\noindent \textbf{(2) Quantum code families admitting constant-depth unitary gates.}
Under mild assumptions, this paper provides general criteria for the existence of a cup product on quantum codes and the conditions for being able to construct an appropriate cohomology invariant. 
We also found some examples of codes and code families that would satisfy these conditions, but it is clear that there is a significant unexplored space of codes that could meet these conditions, which presents an interesting direction for future work. In addition, it is an open question whether there exists a constraint on code parameters (including check weights) such that the gate code admits constant-depth unitary implementation of a gate at $\Lambda$-th level of the Clifford hierarchy, and what code families saturate the constraint, if it exists. 

\noindent \textbf{(3) Bravyi-K\"onig theorem for general qLDPC codes and bounds on fault-tolerance overheads.}
For more general qLDPC codes, a result analogous to the Bravyi-K\"onig theorem~\cite{Bravyi_2013} is not known. 
It would be interesting to see if the picture proposed in this paper can be further developed to derive such a result. As one indicator, a gate at level $\Lambda$ of the Clifford hierarchy is obtained in quantum codes that are products of $\Lambda$ many classical codes.

\noindent \textbf{(4) Quantum phases of matter and qLDPC codes.} 
Topological codes have a deep connection with topological phases of matter. In particular,  stabilizers of a topological code define the local terms of the Hamiltonian providing a fixed-point realization of a topological phase, and the logical subspace of a topological code corresponds to the ground state subspace of this local Hamiltonian. From this viewpoint, transversal gates can be understood in terms of  (emergent) symmetries of the ground state subspace. Such symmetries are part of a generalized notion of symmetries, \cite{Gaiottoetal2015} including higher-form symmetries and higher-group symmetries~\cite{kapustin2015higher}. In this context, the word ``higher" refers to the fact that the symmetries can act on a submanifold of the total space. In the case of higher group symmetries, the symmetries acting on different dimensions can even have interesting commutation relations.

From this viewpoint, the most general viewpoint symmetries corresponds to sweeping a topological defect across the system, as it leaves the system invariant. For the purposes of quantum codes, we are interested in invertible topological defects, that is, defects that have an inverse. This results in a process where we may locally create a defect and its inverse from the vacuum and have them pairwise cancel in the bulk. This can pictorially be thought of as ``pumping" such an invertible defect across the system. In fact, many of these topological defects can be created using a topological action called the Dijkgraaf-Witten action\cite{DijkgraafWitten}, which are classified by the cohomology groups of the gauge group\footnote{These defects can also be understood from the condensed matter viewpoint as the cohomology group of symmetry-protected topological (SPT) phases \cite{Chen13}.}. 
In the literature, there has been ample discussion on how to understand logical gates through pumping of topological defects in topological codes \cite{Yoshida15,Yoshida16,Yoshida17,WebsterBartlett18,Zhu22fractal,Barkeshli23,Chen23,Zhu2023,Kobayashi23,Barkeshli24_1,Barkeshli24_2,song2024magic}.
In particular, these gates can be understood as unitary evolution which traces out non-contractible paths in the space of Hamiltonians~\cite{Aasen22adiabatic}, a concept recently explored in the context of Floquet phases of matter \cite{Potter17,ElseNayak16,RoyHarper17,Tantivasadakarn22pump} and parametrized quantum systems~\cite{Shiozaki22,Wen23Berry,beaudry2023homotopical}.

 However, there are no similar results for non-topological codes. In particular, one might be interested in generalizing the notion of quantum phases of matter to one that relaxes geometric locality but still enforces $k$-locality, which recently became a subject of interest~\cite{Tan2023fracton,rakovszky2023LDPCI,rakovszky2024LDPCII,Stephen24klocal,lavasani2024stability}. We propose to call such phases \emph{homological phases of matter}, which subsumes topological phases. In such setups, only very limited examples of transversal gates corresponding to quantum pumps of symmetry-protected states are known to exist \cite{Kubica18ungauging,TantivasadakarnVijay20,Tantivasadakarn22pump,Hsin23}, and in all such examples, the models are non-topological, but still geometrically local (corresponding to fracton phases of matter). It would be interesting to have a similar classification of ``homological'' actions and understand the corresponding symmetries. This motivates a new generalization of topological quantum field theory (TQFT) with \emph{homological} quantum field theory --- HomQFT. 
 
 A natural place to study this would be to see if such homological actions can be understood from the framework of fiber bundles~\cite{gottesman2013fibre} and if there exists a classification similar to that of Dijkgraaf-Witten theories~\cite{DijkgraafWitten}. The framework of fiber bundles has been particularly fruitful in constructing gates in quantum codes \cite{Conrad2024GKP,Davydova2024DA}. The framework proposed in this paper for constructing gates via cohomology invariants points to the existence of a hierarchy of such symmetries (parametrized by their respective level in the Clifford hierarchy) in $k$-local Hamiltonians. It would be interesting if these symmetries can be described as a $k$-local version of SPT phases and if they can be classified mathematically, possibly by some appropriate cohomology theory.

\section{Simplicial Complexes and Cup Products}

Cohomology operations are constructed naturally when dealing with simplicial cohomology. 
Therefore, we start by reviewing essential definitions for simplicial cohomology and cup product, following the conventions in \cite{steenrod1947products}.

\subsection{Definitions}
A \emph{simplicial complex} $X$ is a collection of non-empty finite subsets, called \emph{simplices}, of some set $S$, such that all non-empty subsets of a simplex $\sigma \in X$ are also simplices in $X.$
A simplex $\sigma\in X$ of cardinality $p-1$ is called a \emph{$p$-simplex}. We denote the set of $i$-simplices in $X$ by $X_i\subset X.$ The $0$-simplices and $1$-simplices in $X$ are called vertices and edges, respectively.

For convenience, we will always assume that the underlying set $S$ of a simplical complex $X$ is equal to the set of vertices $S=X_0$. Moreover, we assume that $X$ is finite. 

We fix a partial order on the vertices in $X$ which induces a total order on each simplex in~$X.$ 
We emphasize that an order locally induces an orientation of each simplex in $X$. It does, however, not induce a global orientation on $X.$ The order is essential to obtain consistent signs in the formulas we discuss below.

\subsection{Simplicial cochain complex}
Recall that for an Abelian group $A,$ the group $C^p(X,A)$ of $A$-valued $p$-cochains on $X$ is defined as the group of $A$-valued functions on $X_p.$ In the following, we will freely identify the following:
\begin{enumerate}
    \item A $p$-simplex $\sigma$.
    \item The ordered array $[v_0,\dots,v_p]$ of the vertices of $\sigma.$
    \item The indicator function of $\sigma$ in $C^p(X,A).$
\end{enumerate}
Since $X$ is finite, the indicator functions form an $A$-basis of $C^p(X,A).$

The simplicial cochain complex associated to $X$ is defined by

\[C^\bullet(X,A) = \biggl\{
\begin{tikzcd}
	{C^0(X,A)} & {C^1(X,A)} & {C^2(X,A)} & \cdots
	\arrow["{\delta^0}", from=1-1, to=1-2]
	\arrow["{\delta^1}", from=1-2, to=1-3]
	\arrow["{\delta^2}", from=1-3, to=1-4]
\end{tikzcd}
\biggr\}
\]
where the $p$-th coboundary operator has the formula
$$\delta^p(f)([v_0,\dots,v_{p+1}])=\sum_{i=0}^p(-1)^if([v_0,\dots, \hat{v_i},\dots,v_{p+1}])$$
for a $f \in C^p(X,A)$ and $\hat v_i$ means that the argument $v_i$ is skipped. Coboundary operations fulfill $\delta^{p+1}\delta^p=0$. We often omit the indices and simply write $\delta$ for the coboundary operator.
We denote by $B^p(X,A)=\im(\delta^{p-1})$ and $Z^p(X,A)=\ker(\delta^p)$ the groups of \emph{$p$-coboundaries} and \emph{$p$-cocycles} on $X$, respectively. The $p$-th \emph{cohomology group} of $X$ is the quotient $H^p(X,A)=Z^p(X,A)/B^p(X,A).$

\subsection{Cup product}\label{sec:cupproduct}
Now assume that $A=R$ is a commutative ring. The cup product is an associative but usually non-commutative product $\cup$ on the complex $C^\bullet(X,R)$ which is defined by
$$(f\cup g)([v_0,\dots,v_{p+q}])=f([v_0,\dots,v_p])g([v_p,\dots,v_{p+q}])$$
for $f\in C^p(X,R)$ and $g\in C^q(X,R).$
For a $p$- and $q$-simplices, the only non-trivial cup product arises for $[v_0,\dots, v_p]$ and $[v_p,\dots,v_{p+q}]$ such that $[v_0,\dots, v_{p+q}]$ is a ($p+q$)-simplex, where
$$[v_0,\dots, v_p]\cup [v_p,\dots,v_{p+q}]=
    [v_0,\dots, v_{p+q}].$$
We emphasize that the cup product very much depends on the ordering of the vertex set.
\begin{example}
    Two edges of a cell complex have a non-zero cup product if and only if the end-vertex of the first is the start-vertex of the second and the edges span a $2$-simplex
    % two edges, cup 0
    \begin{equation}\label{eqn:cupex}
        \begin{tikzpicture}[baseline={([yshift=+0ex]current bounding box.center)}]
            \begin{scope}[decoration={markings,mark=at position 0.6 with {\arrow{>}}}] 
                \draw[postaction={decorate}] (0.5,{sqrt(3)/2}) node[above right]{$a$} -- (0,0) node[below left]{$b$};
            \end{scope}
        \end{tikzpicture}
        \cup
        \begin{tikzpicture}[baseline=(current bounding box.center)]
            \begin{scope}[decoration={markings,mark=at position 0.6 with {\arrow{>}}}] 
                \draw[postaction={decorate}] (0,0) node[left] {$b$} -- (1,0) node[right] {$c$};
            \end{scope}
        \end{tikzpicture}
        =
        \begin{tikzpicture}[baseline=(current bounding box.center)]
        \begin{scope}[decoration={markings,mark=at position 0.6 with {\arrow{>}}}] 
            \draw[postaction={decorate}] (0.5,{sqrt(3)/2}) -- (0,0);
            \draw[postaction={decorate}] (0.5,{sqrt(3)/2}) -- (1,0);
            \draw[postaction={decorate}] (0,0) -- (1,0);
            \node[above] at (0.5,{sqrt(3)/2}) {$a$};
            \node[below left] at (0,0) {$b$};
            \node[below right] at (1,0) {$c$};
        \end{scope}
        \end{tikzpicture}
    \end{equation}
\end{example}

\section{Products and Integrals beyond Simplicial Complexes}\label{sec:complexeswithproducts}

For the discussion to be applicable to a broader class of homological (but not necessarily topological) quantum codes\footnote{Which we use to more generally refer to any quantum code specified by a chain complex, as opposed to a slightly more particular term used in Ref.~\cite{bravyi2013homologicalproductcodes}.} (which encompasses all CSS code constructions), one needs to go beyond simplicial complexes. 
In this section, we will introduce the notion of a dg-algebra and axiomatize all necessary properties to have a concept of cup products and integrals. 
In addition, we discuss how standard product constructions of complexes, such as tensor products and balanced products are compatible with cup products and integrals.

\subsection{Definitions}
\begin{definition}
    Let $R$ be a ring. A \emph{cochain complex} $C$ (over $R$) is a sequence of $R$-modules and linear maps
\[\begin{tikzcd}
	\cdots & {C^{p-1}} & {C^p} & {C^{p+1}} & \cdots
	\arrow[from=1-1, to=1-2]
	\arrow["{\delta^{p-1}}", from=1-2, to=1-3]
	\arrow["{\delta^p}", from=1-3, to=1-4]
	\arrow[from=1-4, to=1-5]
\end{tikzcd}\]
such that $\delta^{p}\delta^{p-1}=0$ for all $p$. 
We write $Z^p(C)=\ker(\delta^p)$, $B^p(C)=\op{im}(\delta^{p-1})$ and $H^p(C)=Z^p(C)/B^p(C)$ for cocycles, coboundaries and cohomology of $C.$
A morphism $\phi:C\to D$ of chain complexes over $R$  is a sequence of $R$-linear maps $\phi_p:C^p\to D^p$ such that $\phi_p\delta^p=\delta^p\phi_p.$
\end{definition}
A map of complexes descends to a map between the cohomologies of the complexes. We are interested in complexes with a product structure.
\begin{definition}\label{def:dgalgebra}
    A (non-unital) differential-graded algebra (or \emph{dg-algebra} for short) over $R$ is a cochain complex $C^\bullet$ equipped with a bilinear associative map  
    $$\cup: C^p\times C^q\to C^{p+q},\text{ for all }p,q$$
    called the cup product, such that
    the \emph{Leibniz rule} holds, that is, 
\begin{align}\label{eq:leibnizrule}
    \delta(f\cup g)=\delta(f)\cup g+(-1)^p f\cup \delta(g).
\end{align}
The Leibniz rule implies that the product $\cup$ descends to a cohomology operation.
\end{definition}

\begin{example}\label{ex:topologicaldgalgebras}
    There are two standard examples that we will make use of later:
    \begin{enumerate}
        \item The simplicial cochain complex $C(X,R)$ associated to a simplicial complex admits a dg-algebra, see \Cref{sec:cupproduct}. The same is true for the singular and cubical (see \cite{Chen2023higher}) cochain complex associated to any topological space.
        \item Let $X$ be a CW-complex. Then the cellular cochain complex $C(X,R)$ is a dg-algebra.
    \end{enumerate}
\end{example}
We also need the following notion of bases.
\begin{definition}
    A \emph{basis} of a complex $C$ over $R$ is a set $X=\biguplus_p X_p$ such that each $C^p$ is free over $R$ with basis $X_p.$ We call $C$ \emph{based} if it is equipped with a basis. A morphism of based chain complex is a morphism of chain complexes that preserves the basis.
\end{definition}
We will often write $C(X,R)$ for a based cochain complex to signify the basis $X=\bigcup X_i$ and ring $R$. 
Moreover, we will write $B^p(X,R)$, $Z^p(X,R)$ and $H^p(X,R)$ for coboundaries, cocycles and cohomology. 
For a based cochain complex $C(X,R)$ over $R$ and a map of rings $R\to S$, we obtain a based cochain complex over $S$ via
$C^p(X,S)=C^p(X,R)\otimes_R S$
with the same basis.

\subsection{Tensor products}\label{sec:tensor_products_cup}
Consider two complexes $C$ and $D$ over a commutative ring $R.$ Then we can consider the tensor product complex $C\otimes_RD$ which consists of the terms
$$(C\otimes_RD)^n=\bigoplus_{p+q=n}C^p\otimes_R D^q$$
and has the coboundary map for $x\in C^p$ and $y\in D^q$ defined by
$$\delta(x\otimes y)=\delta(x)\otimes y+ (-1)^px\otimes \delta(y).$$
If both $C$ and $D$ are equipped with an $R$-bilinear product structure, say $\cup$, then we can equip the $C\otimes D$ with a product by the formula 
\begin{equation}\label{eqn:tensor_products_cup}
    (x_1\otimes y_1)\cup(x_2\otimes y_2)=(-1)^{pq}(x_1\cup x_2)\otimes (y_1\cup y_2)
\end{equation}
for $x_i\in C^p$ and $y_i\in D^q$. If $C$ and $D$ are dg-algebras then so is $C\otimes D$, see
\cite[Section~II]{whitney1938products}.

For based cochain complexes $C(X,R)$ and $C(Y,R)$ we often write $C(X\times Y,R)$ for their tensor product which is also based with basis $X\times Y$. Here,
$$(X\times Y)_n=\biguplus_{p+q=n}X_p\times Y_q$$
is a basis of $C^n(X\times Y,R).$

\begin{example}\label{ex:tensor_product_torus}
    Let $X$ be a simplicial subdivision of the circle $S^1$
    \begin{center}
        \begin{tikzpicture}[scale=0.2]
        % Draw the circle with 8 arrow segments
        \foreach \angle in {0,45,...,315} {
          \draw[<-,] 
            (\angle:3cm) arc (\angle:\angle+45:3cm) node[pos=0.6, circle,draw=black, fill=black, inner sep=0pt,minimum size=1pt] {};
        }
        
        \end{tikzpicture}
    \end{center}
    and let $C$ be the associated cochain complex.
    Then the tensor product cochain complex $C\otimes C$ has basis $X\times X$ which is topologically a 2D torus $T^2$ subdivided into squares.
    \begin{center}
        \begin{tikzpicture}[scale=0.6]
            \clip (.2,.2) rectangle (3.8,3.8);
            \foreach \x in {0,1,2,3} {\
                \foreach \y in {0,1,2,3,4} {
                    % vertical
                    \draw[-] (\y,\x) -- (\y,\x+0.5);
                    \draw[<-] (\y,\x+0.5) -- (\y,\x+1);
                    % horizontal
                    \draw[->] (\x,\y) -- (\x+0.5,\y);
                    \draw[-] (\x+0.5,\y) -- (\x+1,\y);
                }
            }
        \end{tikzpicture}
    \end{center}
    Denote $v$ and $e$ the 0-cochain and 1-cochain dual to the corresponding vertex and edge of the chain complex.
    By \Cref{eqn:tensor_products_cup} the cup product between two edges in $X\times X$ can only be non-zero if one is horizontal and the other vertical, i.e.\ if they are of the form $e_1\otimes v_1$ and $v_2\otimes e_2$.
    In that case we have that
    $$e_1\otimes v_1 \cup v_2\otimes e_2 = - (e_1\cup v_2)\otimes (v_1\cup e_2)$$
    which is non-zero if $v_2$ is the end point of $e_1$ and $v_1$ is the start point of $e_2$, i.e. if we take the cup product between two consecutively oriented edges belonging to a square face in $X\times X$.

    We remark that further products of the from $X \times X\times X \times \cdots$ gives an expression for the cup product on the hypercubic lattice. See \cite{Chen2023higher} for further discussion.

\end{example}

\subsection{Balanced products}

\subsubsection{Invariants and coinvariants}
\label{sec:invariantsandcoinvariants}

Let us first establish some useful concepts which will allow us to extend cup products to the case of balanced products of complexes. 

Let $G$ be a finite group acting on an $R$-module $M$ by $R$-linear maps. Then we denote  $G$-invariants and $G$-coinvariants of $M$ by
\begin{align*}
    M^G=\{m\in M\mid gm=m\}\ \text{ and }\ M_G=M/\langle gm-m\mid m\in M, g\in G\rangle.
\end{align*}
We denote elements in $M_G$ by $[m]_G$ for a representative $m\in M.$
There are natural projection and averaging maps comparing invariants and coinvariants
\begin{align*}
    \op{pr}&\colon M^G\to M_G,\, \quad  m\mapsto [m]_G \ \text{ and }\\
    \op{avg}&\colon M_G\to M^G,\, \quad  [m]_G\mapsto \sum_{g\in G}gm.
\end{align*}
The composition of these maps is multiplication by $|G|.$ 
If $|G|$ is invertible in $R$, then invariants and coinvariants are isomorphic.

There is also another important case where invariants and coinvariants are isomorphic.
Let $M$ be a free $R$-module with a basis $X$ such that $G$ preserves $X$ and acts freely on it. Then the map $\op{avg}$ is an isomorphism. 
Moreover, the  set of orbits $X/G=\{G\cdot x\mid x\in X\}$ indexes a basis of $M_G$ and $M^G,$ where an orbit $G\cdot x$ corresponds to
$[x]_G\in M_G$ and $\sum_{g\in G}g\cdot x\in M^G.$

Similarly, if $G$ acts on a chain complex $C$ by morphisms of chain complexes, then we can consider complexes $C^G$ and $C_G$ of invariants and coinvariants which have terms $(C^G)^p=(C^p)^G$ and $(C_G)^p=(C^p)_G$ with the induced differential. 

If $C(X,R)$ is a based complex and $G$ acts freely on the basis, then the invariants and coinvariants are again based complexes with a basis indexed by the orbits $X/G$. In this case, we take the freedom to identify the invariants an coinvariants via the averaging map and write 
$$C(X/G,R)=C(X,R)_G=C(X,R)^G,$$
\hl{(note that the averaging map $\op{avg}$ is an isomorphism under these assumptions).}

If $C$ is equipped with a product structure $\cup$ preserved by the group action (so ${(g-\cup g-)}={g(-\cup -)}$ for all $g\in G$) then the product structure preserves the invariants~$C^G.$ 
The coinvariants~$C_G$ do not inherit the product structure a priori. 
But, if $C=C(X,R)$ is a based complex and $G$ acts freely on the basis, then we can identify invariants and coinvariants. So both inherit the product structure on $C.$ 
Explicitly, on coinvariants, the cup product can be formulated as
\begin{align}\label{eqn:cupproductcoinvariants}
    [m]_G\cup [n]_G =\sum_{g\in G}\left[m\cup gn\right]_G=\sum_{g\in G}\left[ gm\cup n\right]_G.   
\end{align}
This formula arises from the following computation
\begin{align*} 
[m]_G \cup [n]_G &= 
\op{avg}^{-1}(\op{avg} [m]_G \cup \op{avg}[n]_G) \\ &= 
\op{avg}^{-1}(\sum_{g\in G}g\left(m \cup g^{-1}\op{avg}[n]_G\right)) \\ &= 
\op{avg}^{-1}(\sum_{g\in G}g\left(m \cup \op{avg}[n]_G\right)) \\ &= 
\op{avg}^{-1}(\op{avg}([m \cup \op{avg}[n]_G]_G) = 
\sum_{g\in G}[m \cup gn]_G.
\end{align*}

\subsubsection{Balanced products of chain complexes with cup products}
We will now recall the balanced product construction for chain complexes \cite{James1984,tom2008algebraic,balancedproductcodes,qldpc_review}. 
Let~$X$ and~$Y$ be sets with a right and left $G$-action, respectively. 
The Cartesian product $X\times Y$ admits a diagonal left $G$-action via 
$$g\cdot (x,y)=(x\cdot g^{-1},g\cdot y)\text{ for }g\in G, x\in X \text{ and }y\in Y.$$
The balanced product of $X,Y$ over $G$ is defined as the quotient of the Cartesian product by this diagonal action
$$X\times_G Y=(X\times Y)/G.$$
For an equivalence class $[x,y]_G\in X\times_G Y$ we have $[x\cdot g,y]_G=[x,g\cdot y]_G.$

Similarly, if $C$ and $D$ are complexes with a right and left action of the group $G,$ we obtain a diagonal left action on the tensor product $C\otimes D$. The \emph{balanced product complex} of $C$ and $D$ over $G$ is defined as the coinvariants with respect to the diagonal action
$$C\otimes_GD=(C\otimes D)_G.$$
One may also consider the following variant of the balanced product defined via invariants
$$C\otimes_G^! D=(C\otimes D)^G.$$
The discussion about comparisons between invariants and coinvariants from \Cref{sec:invariantsandcoinvariants} applies to these two definitions.

If $C$ and $D$ are dg-algebras with a multiplication compatible with the group action, then $C\otimes^GD$ inherits the dg-algebra structure of the tensor product $C\otimes D.$

Assume that $C(X,R)$ and $C(Y,R)$ are based complexes with a right and left action of $G$ preserving the bases. Then the balanced product has a basis $X\times_GY$, and we take the freedom to abbreviate
$$C(X,R)\otimes_GC(Y,R)=C(X\times_G Y,R).$$

Additionally assume that $C(X,R)$ and $C(Y,R)$ admit a product $\cup$ such that the group~$G$ preserves the product structure and acts freely on $X\times Y$. 
Then the balanced product ${C(X\times_G Y,R)}$ inherits a product defined by
\begin{align}\label{eqn:cupproductbalancedproduct}
\begin{split}
       [x_1,y_1]_G\cup [x_2,y_2]_G&=(-1)^{qr}\sum_{g\in G}[x_1\cup x_2\cdot g^{-1},y_1\cup g\cdot y_2 ]_G\\   &=(-1)^{qr}\sum_{g\in G}[x_1\cdot g^{-1}\cup x_2, gy_1\cup y_2 ]_G. 
\end{split}
\end{align}
where $y_1$ and $x_2$ have degrees $q$ and $r$, respectively, and we used \Cref{eqn:cupproductcoinvariants}.

\hl{Similarly, for a triple balanced product,} 
$$\mathcolorbox{C(X,R)\otimes_GC(Y,R)\otimes_GC(Z,R)=C(X\times_G Y\times_G Z,R)}$$
\hl{we obtain the following formula for a triple cup product}
\begin{align} \label{eq:triple-cup}
    \begin{split}
       &\mathcolorbox{
       [x_1,y_1,z_1]_{G\times G}\cup [x_2,y_2,z_2]_{G\times G}\cup [x_3,y_3,z_3]_{G\times G}} \\=
       &\mathcolorbox{
       \epsilon \sum_{\mathclap{g,h,g',h'\in G}}[x_1g^{-1}\cup x_2\cup x_3g'^{-1},gy_1h^{-1}\cup y_2\cup g'y_3h'^{-1},hz_1\cup z_2\cup h'z_3]_{G\times G}}
    \end{split}
\end{align}
\hl{where $\epsilon$ is a sign depending on the degree of the $x_i,y_i$ and $z_i$.}

\subsection{Integrals}\label{sec:integrals}
An integral of dimension $n$ on a cochain complex $C$ over $R$ is linear map of the form
$$\int: C^n\to R$$
such that $\int b = 0$ for all $n$-coboundaries $b\in B^n(C).$ 
In particular, an integral induces a map $\int: H^n(C)\to R.$ By convention, we set $\int x=0$ if $x\in C^q$ for $q\neq n.$ Later in the text, we will also use the notation $\int_n$ to denote an integral of dimension $n$. 

Now let $C$ and $D$ be complexes with integrals $\int_C$ and $\int_D$ in dimensions $n$ and $m$ respectively. 
Moreover, assume that $C^i=D^j=0$ for $i>n$, $j>m$ and $i,j<0.$ 
Then we obtain an integral of dimension $n+m$ via
\begin{align}\label{eqn:integral_tensor_product}
    \int: (C\otimes D)^{n+m}\to R, \quad  x\otimes y\mapsto \int_Cx\int_Dy.
\end{align}
One also obtains an integral on the balanced product $C\otimes_G D$ if the group action preserves the integrals on $C$ and $D.$

For a simplicial complex $X$ an integral can be constructed geometrically. 
Assume that the geometric realization of $|X|$ of $X$ is an $n$-dimensional orientable manifold. 
Then any orientation on $|X|$ induces an $n$-dimensional integral on~$X.$ 
The integral of an $n$-dimensional simplex $\sigma\in X$ is $\pm 1$, and the sign can be determined as follows. 
The simplex $\sigma$ has two orientations:  the first comes from the order of the vertices in $X$, and the second is induced by the orientation of $|X|.$ 
Then, $\int\!\sigma$, is positive or negative if the two orientations agree or not.
\section{Classical and Quantum Codes}\label{sec:quantumcodes}
In this section, we discuss the connection between based chain complexes and classical and quantum codes.

First, consider a based cochain complex consisting of at least two terms (recall that  $C^p$ has a basis $X_p$):
    % https://q.uiver.app/#q=WzAsNSxbMSwwLCJDXntpLTF9Il0sWzIsMCwiQ15pIl0sWzMsMCwiQ157aSsxfSJdLFs0LDAsIlxcY2RvdHMiXSxbMCwwLCJcXGNkb3RzIl0sWzAsMSwiZF57aS0xfSJdLFsxLDIsImReaSJdLFsyLDMsImRee2krMX0iXSxbNCwwXV0=
\[\begin{tikzcd}
	\cdots & {C^{p-1}} & {C^p} & \cdots
	\arrow[from=1-1, to=1-2]
	\arrow["{\delta^{p-1}}", from=1-2, to=1-3]
	\arrow[from=1-3, to=1-4]
\end{tikzcd}\]
To define a classical code, we can identify the elements of $X_p$ with bits and the elements of $X_{p-1}$ with checks of the code, so that $\delta^\top$ is the parity check matrix.\footnote{It follows that the code space is exactly given by cycles $Z_p(X,\Zn{2})$.}

A CSS quantum code is defined by two classical codes with parity check matrices~$H_X$ and~$H_Z$ that fulfill the relation $H_Z H_X^\top = 0$.
Identifying $\delta^{p-1}=H_X^\top$ and $\delta^p=H_Z$, a CSS quantum code is equivalent to three consecutive terms in a based cochain complex:
    % https://q.uiver.app/#q=WzAsNSxbMSwwLCJDXntpLTF9Il0sWzIsMCwiQ15pIl0sWzMsMCwiQ157aSsxfSJdLFs0LDAsIlxcY2RvdHMiXSxbMCwwLCJcXGNkb3RzIl0sWzAsMSwiZF57aS0xfSJdLFsxLDIsImReaSJdLFsyLDMsImRee2krMX0iXSxbNCwwXV0=
\[\begin{tikzcd}
	\cdots & {C^{p-1}} & {C^p} & {C^{p+1}} & \cdots
	\arrow[from=1-1, to=1-2]
	\arrow["{\delta^{p-1}}", from=1-2, to=1-3]
	\arrow["{\delta^p}", from=1-3, to=1-4]
	\arrow[from=1-4, to=1-5]
\end{tikzcd}\]
Qubits are identified with elements in $X_p$, and for an $[[n,k,d]]$ quantum code, we have $n=|X_p|$.
The coboundary operator~$\delta^p$ defines two distinguished subspaces of $p$-cochains $C^p$ called \emph{cocycles} $Z^p = \ker \delta^p$ and \emph{coboundaries} $B^p = \im \delta^{p-1}$, as well as the cohomology groups $H^p = Z^p/B^p$.

We can write down the codestates explicitly.
Consider a code defined on qubits and denote by $\ket{0}$ and $\ket{1}$ the standard  basis of $\C^2$.\footnote{Also called the \emph{$Z$-basis} as the Pauli-$Z$ operator is diagonal in this basis.}
The standard (``computational'') basis of the $n$-fold tensor product Hilbert space $\mathcal H = (\C^2)^{\otimes n}$ has basis vectors 
$$\ket{c_1,\dots, c_n}=\ket{c_1}\otimes \cdots \otimes\ket{c_n} \text{ for }c_i\in\{0,1\}.$$
For a fixed basis $x_1,\dots,x_n \in X_p$, we write $\ket{c}=\ket{c_1,\dots, c_n}\in(\C^2)^{\otimes n}$ for $c=\sum_ic_ix_i\in C^p(X,\Zn{2})$; hence, it is labeled by the $p$-cochains.

We can now write down the logical states of the CSS quantum code explicitly.
We will write $B^p$, $Z^p$ and $H^p$ for boundary, cycle and cohomology groups to ease notation.
The code space $\mathcal{C}\subseteq (\C^2)^{\otimes n}$ is spanned by elements of the form
$$\ket{[\gamma]}=\frac{1}{\sqrt{|B^p|}}\sum_{b\in B^p}\ket{\gamma+b}$$
where $\gamma \in Z^p$ is a representative of a cohomology class $[\gamma]\in H^p=Z^p/B^p$.
In terms of the stabilizer formalism, we have that $X$-stabilizers are exactly given by operators of the form~$X^{\otimes b}$ for $b\in B^p$ and the $X$-checks are $X^{\otimes \delta x}$ for $x\in X_p$.\footnote{A similar discussion applies to $Z$-stabilizers in terms of boundaries and boundary operators. We refer the reader to~\cite{qldpc_review} as a reference.}

\subsection{Cohomology operations and logical gates} \label{sec:operations-gates}

The \emph{cohomology invariants}, by definition, are maps that assign numbers which (we can map to phases in $U(1)$) to cohomology classes $[\gamma]\in H^p$.  Because the code states of a quantum code are labeled by the cohomology classes, any such cohomology invariant is a quantum gate. For example, consider a cohomology invariant given by $\psi: H^p \rightarrow \Zn{m}$. A cohomology class  $[\gamma]\in H^p$ is then put in correspondence with a phase $e^{2 \pi i \psi{[\gamma]}/m}$. 

In order for a logical gate to be implemented at a circuit level, we restrict our consideration to operations $\psi$ on $p$-cochains of the form
$$\psi: C^p(X, \Zn{2}) \rightarrow \Zn{m}$$
that \emph{descend} to a cohomology invariant $\psi: H^p \rightarrow \Zn{m}$. The latter defines its logical action.

Since we associated the basis in $C^p(X, \Zn{2}) $ with qubits, such an operation is well-defined on the physical level.
If we write the states in the physical Hilbert space  $\mathcal{H}$ as
$\ket{c}$ for $c \in C^{p}(X,\Zn{2})$, as discussed above, the physical circuit corresponding to $\psi$ is defined such that its action satisfies 
$$U_{\psi} \ket{c} = e^{2 \pi i \psi({c})/m} \ket{c}.$$
Because we demanded that $\psi$ is a cohomology invariant, when $c$ is a cocycle, we have $\psi({c} )\equiv \psi([c])$, and thus, the action of the gate $U_\psi$ depends only on the cohomology class of $c$ -- in other words, it depends only on the logical state. Thus, $U_\psi$ is a logical gate that preserves the codespace. At the same time, because the definition above is written for any basis state in the Hilbert space, it can be realized by a circuit as long as the circuit has this action. 

Finally, if the operation $\psi(c)$, as defined at the level of $p$-cochains, is ``local'', there exists a finite-depth circuit implementation of $U_\psi$. We informally define the operation on $p$-cochains to be local if its value depends only on (a) values of $p$-cochains and (b)  membership of cochains in the boundary maps (which we assume to be finite-degree, so they are ``local'') of other $p$-cochains. The operations that we construct in the rest of this paper are, in fact, integrals over products of a few $p$-cochains. Such an operation is naturally ``local'' in the sense above. Under this condition, $U_\psi$ can be decomposed into a product of phase factors that involve a few qubits at a time only, and each qubit appears in finitely many factors, which makes the circuit finite-depth and $k$-local (or local, if we consider a manifold). 

As a remark, one can consider $\Lambda$ copies of the same quantum code and define a multilinear operation $\Psi_\Lambda$ that similarly maps combinations of cohomology classes $H^p \times \dots \times H^p$ to numbers. This will be to construct more interesting examples of gates from cup products (which are multiple-argument and multilinear) than we would be able to otherwise. This is a straightforward generalization of the discussion here, which we present in \cref{sec:highercontrolz}.  

Moving forward, we will assume that the qubits are associated with the elements in $X_1$, hence $p = 1$. We will also consider only the case where $m = 2$, and therefore $e^{i\phi([\gamma])} \equiv (-1)^{\psi([\gamma])}$.

The goal of the rest of the paper is to show how to build cochain operations $\Psi$ with these properties that also descend to cohomology invariants on tensor and balanced product codes (which we do in the next section) and study resulting gates and circuits (which we do in \cref{sec:highercontrolz}).

\section{Conditions for the Cup Product Structure on Non-Topological Codes}
\label{sec:cup-non-topological}

The approach that we use to construct cochain operations that descend to cohomology invariant is based on first starts with constructing operations with certain properties on classical codes. Interesting quantum codes can be constructed from tensor products\footnote{Also called hypergraph products in the context of quantum codes.} and balanced products of classical codes. 
For these quantum codes, we show that the cochain operations on constituent classical codes give rise to a cohomology invariant on the quantum code if certain conditions on these operations are fulfilled.

\subsection{Products, integrals and orientations}\label{sec:productsforclassicalcodes}

Recall from the previous sections that a classical code $C$ can be interpreted as a based chain complex of the form
$$C^0(X,\Zn{2})\stackrel{\delta}{\to} C^1(X,\Zn{2}).$$
The bases $X_0$ and $X_1$ label the checks and physical bits of the code, respectively, and the differential $\delta$ is the transpose of the parity check matrix.  
The discussion in this section builds on \Cref{sec:complexeswithproducts}, and we do not assume that $C$ is a simplicial complex anymore. 

In this subsection, we determine the conditions for the existence of ($\Lambda$-fold) products and integrals on classical codes, and will formulate a condition on them that will allow us to proceed with constructing cohomology invariants (leading to logical gates) on quantum codes obtained as a tensor or balanced product of several classical codes.  

First, we note that defining an integral is simple. If we assume that the support of every check is even, then we can define a $1$-dimensional integral $\int_1$, see \Cref{sec:integrals}, on the complex which maps every bit to one, $\int_1 x=1$ for $x\in X_1.$

To define a product structure, we note that the cup product between a vertex and an edge in an oriented graph depends on the orientation of the graph. Inspired by this, we introduce a notion that is somewhat similar to orientation.

\begin{definition} \label{def:pre-orientation}
A \emph{pre-orientation} $\mathcal{O}$ of $C(X,\Zn{2})$ is a collection of partitions
$\delta(a)=\din(a)\uplus \dout(a)\uplus \delta_{\op{free}} (a)$
of the coboundary of every $a\in X_0$.
\end{definition}

In this definition, and in the following, we identify a vector with its support.

\begin{definition} \label{def:cup}
    Given a pre-orientation $\mathcal{O}$, we define a bilinear cup product $\cup_{\mathcal{O}}=\cup$ on the complex $C(X,\Zn{2})$ via 
    \begin{enumerate}
        \item[(1)] $a\cup a=a$ for $a\in X_0,$
        \item[(2)] $a\cup x = x$ for $a\in X_0$ and $x\in \dout(a)$,
        \item[(3)] $x\cup a = x$ for $a\in X_0$ and $x\in \din(a)$,
    \end{enumerate}
    and is zero for any other combination of elements in $X$ that is not specified above.
\end{definition}

In the following, we will call a multiple-cup product \hl{(i.e. multiple applications of the cup product) of the form $a_1 \cup a_2\cup\dots \cup a_{\Lambda}$ a \emph{$\Lambda$-fold cup product}}. In addition, we remark that, for our purposes, the condition (1) in the definition of cup product could be technically made more general.

\begin{definition}[Non-overlapping bits] \label{def:non-overlapping}
Given a pre-orientation $\mathcal{O}$, we define the following conditions: 
\begin{equation}
\begin{split}\label{eq:non-overlap}
    &\din (a) \cap \din(a') = \varnothing, \text{ and}
    \\
    &\dout (a) \cap \dout(a') = \varnothing.
    \end{split}
\end{equation}
for every $a \neq a' \in X_0$.  Here, $\cap$ stands for the intersection (and not the cap product)
\end{definition}
In other words, the ``non-overlapping bits'' condition demands that no two checks can share the same incoming bit, nor the same outgoing bit. Such a condition is automatically fulfilled, for example, for an oriented graph.

\begin{proposition}
For a classical code with preorientation satisfying the no overlapping bits condition, the $\Lambda$-fold cup product obtained by multiple applications of the cup product defined in \cref{def:cup} is associative.
\end{proposition}
\begin{proof}
    Obtained by direct verification.
\end{proof}
 
\begin{remark}
    In the following, unless otherwise specified, we will \emph{not} impose the non-overlapping bits condition, which means that the cup product need not be associative.
    \hl{For example, $(a\cup a')\cup x$ is in general not equal to to $a\cup (a'\cup x)$ for $a,a'\in X_0$ and $x\in X_1$.
    One can choose the order of operations that depends on the degree of the arguments in the cup product, such that one never takes a cup product between two elements in $X_0$, namely:} 
    \begin{align}\mathcolorbox{a_0\cup\dots \cup a_{j-1}\cup x\cup a_{j+1}\cup\dots\cup a_\Lambda=\begin{cases}
        x & \text{if } x\in \bigcap_{i=1}^{j-1}  \dout(a_i)\cap \bigcap_{k=j+1}^{\Lambda}  \din(a_k),\\
        0 &\text{otherwise.}
    \end{cases}}
    \end{align}
\end{remark}
\hl{Moving forward, we will assume the order of operations as in the equation above, which is the most symmetric choice. We note that in the earlier version of the paper, we implicitly assumed this order, which is consistent with equation (\ref{eq:leibnizcondition}), whereas Eq.~(\ref{eq:lambda3}) contained a mistake that has been fixed in the current version\footnote{We thank Ryan Tiew for useful discussions and remark that Ref.~\cite{menon2025magictricyclesefficientmagic} independently emphasized on this order of operations. }.  Fixing a different order of operations would also be valid, but would lead to a different set of conditions on the chain complex and a different action of the logical gate as a result, which would be interesting to study systematically. For example, different choices were made in Refs.~\cite{menon2025magictricyclesefficientmagic} and \cite{2508.08191,li2025transversaldimensionjumpproduct}. }

Generically, a cup product defined this way does not satisfy the Leibniz rule without imposing additional conditions on the chain complex.   With our end goal of constructing a logical gate on product complexes from cup products on constituent complexes, it is sufficient if the Leibniz rule is satisfied only after applying the integral. Namely, we define

\begin{definition}[Integrated Leibniz rule] \label{def:int-Leibniz}
For a chain complex $C$ with an integral and a $\Lambda$-fold cup product. If the condition 
\begin{equation}
    \int_1 \sum_{j=1}^\Lambda a_1\cup\dots \cup \delta(a_j)\cup \dots \cup a_\Lambda=0
 \ (  \op{mod} 2)
\end{equation} 
\noindent is satisfied (recall that an integral of a coboundary is defined to be zero), then we say that the $\Lambda$-fold cup product obeys the \emph{integrated Leibniz rule}. Here, $a_i$ could be either in $C^0$ or $C^1$ for any $i \in [1,\Lambda]$.
\end{definition}

Imposing these weaker conditions allows us to later construct gates for a wider class of codes than one would have been able to otherwise\footnote{In addition, if one instead considered the stronger condition that the Leibniz rule holds without an integral and also assumed that $\delta_{\op{free}} = 0$, this would be possible to satisfy only on graphs.}. This approach equips the complexes that we consider with a ``weak'' dg-algebra based on these weaker conditions. 

Below, we determine the conditions on a chain complex (or, equivalently, for a classical code) under which the integrated Leibniz rule holds for a $\Lambda$-fold cup product.
\begin{proposition}  \label{prop:integrated-leibniz-conditions}
Consider a chain complex $C$ with a pre-orientation $\mathcal O$, a 1-dimensional integral, and a $\Lambda$-fold cup product from \cref{def:cup}.  The $\Lambda$-fold cup product satisfies the integrated Leibniz rule in \cref{def:int-Leibniz} if for all $a_i \in X_0$:
\begin{equation}\label{eq:leibnizcondition}
    \sum_{j=1}^\Lambda \left | \left ( \bigcap_{i=1}^{j-1}  \dout(a_i) \right ) \cap \delta(a_j) \cap \left ( \bigcap_{k=j+1}^{\Lambda}  \din(a_k) \right )  \right | = 0 \ (\op{mod} 2).
\end{equation}
Here, if the limit of the series of intersections is less than 1 or larger than $\Lambda$, the respective operation is assumed to be removed. 

As a consequence:
\begin{itemize} 
    \item [(a)]for $\Lambda = 2$, the condition becomes
 \begin{equation}\label{eq:lambda2}
 \begin{split}
    &\left | \din(a_1)\right | +\left | \dout(a_1)\right |   = 0 \ (  \op{mod} 2), \quad \text{and} 
    \\
    &\left | \din (a_1) \cap \din(a_2)\right |+ 
    \left | \dout(a_1) \cap \dout(a_2)\right | 
    \\
    &+ \left | \dout(a_1) \cap \delta_{\op{free}}(a_2)\right | +\left | \delta_{\op{free}}(a_1) \cap \din(a_2) \right |   = 0 \ (  \op{mod} 2), \  \quad a_1 \neq a_2. 
    \end{split}
 \end{equation}
 $\forall a_1, a_2 \in X_0$.

 \item [(b)]for $\Lambda = 3$, the condition becomes:
\begin{align} 
    &\mathcolorbox{|\din(a_1)|+|\dout(a_1)|=0\ (\op{mod} 2)} \nonumber \\
    &\mathcolorbox{|\din(a_1)\cap \din(a_2)|+|\dout(a_1)\cap \dout(a_2)|=0\ (\op{mod} 2)}\nonumber \\
    &\mathcolorbox{|\dout(a_1)\cap \dfree(a_2)|=0\ (\op{mod} 2)} \nonumber \\
    &\mathcolorbox{|\dfree(a_1)\cap \din(a_2)|=0\ (\op{mod} 2)} \label{eq:lambda3} \\
    &\mathcolorbox{|\din(a_1) \cap \din(a_2) \cap \din(a_3)| + |\dfree(a_1) \cap \din(a_2) \cap \din(a_3)|} \nonumber\\
    &\mathcolorbox{+ |\dout(a_1) \cap \dfree(a_2) \cap \din(a_3)| + |\dout(a_1) \cap \dout(a_2) \cap \dfree(a_3)|} \nonumber\\
    &\mathcolorbox{+ |\dout(a_1) \cap \dout(a_2) \cap \dout(a_3)|=0 \ (\op{mod} 2)} \nonumber
\end{align} 

 $\forall a_1, a_2, a_3  \in X_0$. 
\end{itemize}

\end{proposition}
\begin{proof}
    Obtained by directly applying the decomposition in \cref{def:pre-orientation} and the definition of the cup product, \cref{def:cup}. 
\end{proof}

We note that the condition for $\Lambda =2 $ is automatically satisfied if the condition for $\Lambda \geq 3$ is fulfilled.
The conditions for the integrated Leibniz rule become stronger as $\Lambda$ increases, which is very different from the case with the regular cup products for which the usual Leibniz rule holds. Namely, if the cup product is associative and fulfills the usual (not integrated) Leibniz rule, the Leibniz rule automatically follows for the $\Lambda$-fold cup product for every $\Lambda$.

The conditions above were formulated without requiring the ``non-overlapping bits'' condition. If this is imposed, then the cup product is associative and the conditions simplify to:

\begin{itemize} 
    \item [(a)]for $\Lambda = 2$:
 \begin{equation}\label{eq:lambda21}
 \begin{split}
    &\left | \din(a_1)\right | +\left | \dout(a_1)\right |   = 0 \ (  \op{mod} 2), \text{ and} 
    \\
    &\left | \dout(a_1) \cap \delta_{\op{free}}(a_2)\right | +\left | \delta_{\op{free}}(a_1) \cap \din(a_2) \right |   = 0 \ (  \op{mod} 2), \ a_1 \neq a_2. 
    \end{split}
 \end{equation}
 $\forall a_1, a_2 \in X_0$.

 \item [(b)]for $\Lambda = 3$, the condition becomes independent of $\Lambda$ and is:
 \begin{equation} \label{eq:lambda31}
 \begin{split}
    &\left | \din(a_1)\right | +\left | \dout(a_1)\right |   = 0 \ (  \op{mod} 2),
    \\
    &\left | \delta_{\op{free}}(a_1) \cap \din(a_2)  
 \right |   = 0 \ (  \op{mod} 2), \ a_1 \neq a_2 ,
 \\
 &\left |\dout(a_1) \cap  \delta_{\op{free}}(a_2)  
 \right |   = 0 \ (  \op{mod} 2), \ a_1 \neq a_2, \\
 &\left | \dout (a_1) \cap \delta_{\op{free}}(a_2) \cap \din(a_3)  
 \right |   = 0 \ (  \op{mod} 2), \ a_1 \neq a_2 \neq a_3, \ a_1 \neq a_3.
    \end{split}
 \end{equation}

 $\forall a_1, a_2, a_3  \in X_0$. 
\end{itemize}

One possible way to find codes that can have an appropriate pre-orientation is to find classical codes where there exists a subset of checks and bits that one can consider that together form a graph inside the code (while \emph{not} being disjoint from the rest of the code). We make this precise in the following.

\begin{proposition} \label{prop:hidden-graph}
    Consider a complex $C(X,\Zn{2})$ such that there is a subset $Y_1\subset X_1$ such that $X_0$ and $Y_1$ form a graph with respect to the incidence relation given by the differential~$\delta$ of $C(X,\Zn{2})$. We denote the resulting graph by $\mathcal{G}.$
    Then we choose an orientation on $\mathcal{G}$ such that the sets of incoming and outgoing edges of each vertex have the same parity. Now, we define a pre-orientation $\mathcal{O}$ on $C(X,\Zn{2})$ via
    \begin{align*}
        \delta_{\op{in}}(v)&=\{y\in Y_1\mid y \text{ is an incoming edge incident to } v \text{ in the graph } \mathcal{G}\}\\
        \delta_{\op{out}}(v)&=\{y\in Y_1\mid y \text{ is an outgoing edge incident to } v \text{ in the graph } \mathcal{G}\}\\
        \delta_{\op{free}}(v)&= \delta(v)\setminus Y_1.
    \end{align*}
    Then the pre-orientation $\mathcal{O}$ fulfills conditions \Cref{eq:non-overlap} and \Cref{eq:leibnizcondition}. Hence, the cup product defined via the pre-orientation $\mathcal{O}$ is associative and fulfills the $\Lambda$-fold integrated Leibniz rule for all $\Lambda\geq 1.$
\end{proposition}
\begin{proof}
    The non-overlapping condition \Cref{eq:non-overlap} comes from the oriented graph $\mathcal{G}$. Then condition \Cref{eq:leibnizcondition} boils down to \Cref{eq:lambda31}. This is fulfilled by the assumption on the orientation on $\mathcal{G}$ and since $\delta_{\op{in}}(a)\cap \delta_{\op{free}}(a')\subset Y_1\cap (X_1\setminus Y_1)=\varnothing$ and $\delta_{\op{out}}(a)\cap \delta_{\op{free}}(a')\subset Y_1\cap (X_1\setminus Y_1)=\varnothing$, for all $a,a'\in X_0.$
\end{proof}

\subsection{Cup products operation on quantum codes}\label{sec:fromclassicaltoquantumcups}
Now consider a family of complexes/codes: 
$$C(X^{(i)},\Zn{2}): C^0(X^{(i)},\Zn{2})\stackrel{\delta^{(i)}}{\to} C^1(X^{(i)},\Zn{2})$$
for $i=1,\dots,\Lambda$, each equipped with a pre-orientation $\mathcal{O}^{(i)}$. 
We denote the tensor product complex by
\begin{align*}
C(\underline{X},\Zn{2})=C(X^{(1)},\Zn{2})\otimes\dots \otimes C(X^{(\Lambda)},\Zn{2}).
\end{align*}
which we put in correspondence with a quantum code with qubits in degree one. Using \Cref{eqn:tensor_products_cup} from \Cref{sec:tensor_products_cup}, the tensor product is also equipped with a cup product, and, by \Cref{eqn:integral_tensor_product} in \Cref{sec:integrals}, also with an $\Lambda$-dimensional integral.

Let us finally construct an operation which we will base the gates for quantum codes on:

\begin{definition}\label{def:quantumcopgate}
    Given a tensor product complex $C(\underline{X},\Zn{2})$, we define the operation~$\Psi_{\cup, \Lambda}$ to be:
    \begin{equation}
        \Psi_{\cup, \Lambda}: C^{1}(\underline{X},\Zn{2})\times \cdots \times C^{1}(\underline{X},\Zn{2})\to \Zn{2}, \quad  \Psi_{\cup, \Lambda}(\underline{x}_1,\ldots, \underline{x}_\Lambda)=\int_\Lambda \underline{x}_1\cup\cdots\cup\underline{x}_\Lambda
    \end{equation}
   where $\int_\Lambda$ is the integral of dimension $\Lambda$. 
\end{definition}

 Let us show that the integrated Leibniz rule on individual classical codes is sufficient for~$\Psi_{\cup, \Lambda}$ to be a cohomology invariant on a quantum code:

\begin{lemma} [Sufficient condition for a cohomology invariant] \label{lemma:sufficient_cup}
Assume that complexes $C(X^{(i)},\Zn{2})$ for $i = 1, \ldots, \Lambda$ are equipped with pre-orientations $\mathcal O^{(i)}$, and have a 1-dimensional integral each. In addition, assume that the integrated Leibniz rule (\cref{def:int-Leibniz}) is fulfilled for the $\Lambda$-fold cup product up to integration, namely:
\begin{equation}
    \int_1 \sum_{j=1}^\Lambda x_1\cup\ldots \cup \delta(x_j)\cup \ldots \cup x_\Lambda=0.
\end{equation}
where $x_k \in X_0^{(i)}$ for all $k = 1 \ldots \Lambda$. 
This is sufficient for the map $\Psi_{\cup, \Lambda}$ in \cref{def:quantumcopgate} to be a cohomology operation. 
\end{lemma}
\begin{proof}

\hl{Consider a simple tensor product chain $\underline{y} = \bigotimes_{s = 1}^\Lambda y^{(s)}$. By Leibniz rule for the coboundary operation, we can write} $$  \mathcolorbox{\underline{\delta y} = \sum_k \delta_k \underline y,   \quad \text{where} \quad  \delta_k \underline{y} = y^{(1)} \otimes \ldots \otimes \delta_k y^{(k)} \otimes \ldots \otimes y^{(\Lambda)}.}  $$ \hl{and we sum over $k$ from $1$ to $\Lambda$.}  \hl{The cup product of $\Lambda$ such chains can be written as}
\begin{equation}
    \mathcolorbox{\underline{y}_1\cup\ldots  \cup \underline{y}_\Lambda = \bigotimes_{s = 1}^\Lambda \left (  y_1^{(s)} \cup \ldots y_\Lambda^{(s)}  \right).} 
\end{equation}
\hl{Consider, for a fixed $k$:}
\begin{equation}
   \mathcolorbox{\sum_j \underline{y}_1\cup \ldots \cup \delta_k (y_j)\cup \ldots \cup \underline{y}_\Lambda = \sum_j \left (  y_1^{(k)} \cup \ldots \cup \delta_k y_j^{(k)} \cup \ldots y_\Lambda^{(k)}  \right) \bigotimes_{s \neq k} \left (  y_1^{(s)} \cup \ldots y_\Lambda^{(s)}  \right).}
\end{equation}
\hl{We can now integrate this expression using linearity and $\int_\Lambda = \int_1^{(1)}  \ldots  \int_1^{(\Lambda)}$:}
\begin{equation}
\begin{split}
   \mathcolorbox{\sum_j \int_\Lambda \underline{y}_1\cup \ldots \cup \delta_k (y_j)\cup \ldots \cup \underline{y}_\Lambda =}  &  \mathcolorbox{\int_1^{(k)}  \sum_j\left (  y_1^{(k)} \cup \ldots \cup \delta_k \underline y_j^{(k)} \cup \ldots y_\Lambda^{(k)}  \right)} \\
     \bigotimes_{s \neq k} &  \mathcolorbox{\int_1^{(s)}  \left (  y_1^{(s)} \cup \ldots y_\Lambda^{(s)}  \right) = 0}
   \end{split}
\end{equation}
\hl{obtaining zero because by integrated Leibniz rule for the constituent complexes in the product.  By linearity, it also immediately follows that} 
\begin{equation}
\begin{split}
     \mathcolorbox{\sum_j \int_\Lambda \underline{y}_1\cup \ldots \cup \underline{\delta} (y_j)\cup \ldots \cup \underline{y}_\Lambda = \sum_k \sum_j \int_\Lambda \underline{y}_1\cup \ldots \cup \delta_k (\underline y_j)\cup \ldots \cup \underline{y}_\Lambda  = 0.}
   \end{split}
\end{equation} 
  \hl{Additionally, by linearity of all operations involved in the expression above, it follows that the expression $ \sum_j \int_\Lambda \underline{x}_1\cup \ldots \cup \underline{\delta} (x_j)\cup \ldots \cup \underline{x}_\Lambda$ is zero for any choice of $x_i$, and not just pure tensor product ones. This shows the integrated Leibniz rule for the product complex. }

\hl{Now, let us show that $\Psi_{\cup, \Lambda}$ is a cohomology operation on the product complex $C(\underline{X}, \Zn{2})$. Set $\underline{z}_i \in Z^{1}(\underline{X}, \Zn{2})$ for $i = 1, \ldots, \Lambda$ and $\underline{a}_s \in C^0(\underline{X}, \Zn{2})$. Then, for  $\Psi_{\cup, \Lambda}$ to be a cohomology operation it is sufficient if}
\begin{equation}  \label{eq:psicondition1}
\begin{split}
     \mathcolorbox{\Psi_{\cup, \Lambda}(\underline{z}_1,\ldots, \underline{z}_s + \underline{\delta}  (\underline{a}_s), \ldots  \underline{z}_\Lambda)- \Psi_{\cup, \Lambda}(\underline{z}_1,\ldots, \underline{z}_s, \ldots  \underline{z}_\Lambda)=0.}
\end{split}
\end{equation}
\hl{Writing out the expression for $\Psi_{\cup, \Lambda}$ explicitly, we arrive at:}
\begin{equation}  
\begin{split}
     \mathcolorbox{\Psi_{\cup, \Lambda}(\underline{z}_1,\ldots, \underline{z}_s + \underline{\delta}  (\underline{a}_s), \ldots  \underline{z}_\Lambda)- \Psi_{\cup, \Lambda}(\underline{z}_1,\ldots, \underline{z}_s, \ldots )= \int_\Lambda \underline{z}_1 \cup \ldots \cup \underline{\delta}  (\underline{a}_s) \cup \ldots   \cup \underline{z}_\Lambda  }
\end{split}
\end{equation}
\hl{By integrated Leibniz rule for the product complex that we just showed, it follows that}
 \begin{equation}  
\begin{split}
      \mathcolorbox{\int_\Lambda \underline{z}_1 \cup \ldots \cup \underline{\delta}  (\underline{a}_s) \cup \ldots   \cup \underline{z}_\Lambda  = \sum_{k \neq s} \int_\Lambda \underline{z}_1 \cup \ldots \cup \underline{\delta}  \underline{z}_k \cup \ldots \cup  \underline{a}_s \cup \ldots   \cup \underline{z}_\Lambda}
\end{split}
\end{equation}
\hl{which is zero because $\delta z_k = 0$ for all $k$. }
\end{proof}

A similar result applies to balanced product codes. Namely, if there are groups~$G^{(i)}$ acting from the right on $C(X^{(i)},\Zn{2})$ and the left on $C(X^{(i+1)},\Zn{2})$ for each $i=1,\dots, \Lambda-1,$ such that the group actions preserve the bases, then we can consider the balanced product complex 

\begin{equation} \label{eq:balancedproductquantum}
C(\underline{X}/\underline{G},\Zn{2})=C(X^{(1)},\Zn{2})\otimes_{G^{(1)}}\dots \otimes_{G^{(\Lambda-1)}} C(X^{(\Lambda)},\Zn{2}).
\end{equation}

In case of balanced product codes, the pre-orientation will be the same as in \cref{def:pre-orientation}, except we additionally impose the condition that it is preserved under the group action. Namely, for the complex $C(X^{(i)}, \Zn{2})$, the pre-orientation has to be preserved under the right action of group $G^{(i)}$ and the left action of group $G^{(i-1)}$ (apart from $i = 1, \Lambda$ when it needs to be preserved under one group action). We will also assume that the group acts freely on the basis. 
Then, the the balanced product complex inherits the 
cup product (\cref{def:cup}) and the integral from $C(\underline{X}/\underline{G},\Zn{2})$. The operation $\Psi_{\cup, \Lambda}$ straightforwardly extends to the case of balanced product complex:

\begin{definition}\label{def:quantumcopgate_balanced}
     Given a balanced product complex $C(\underline{X}/\underline{G},\Zn{2})$ (\cref{eq:balancedproductquantum}), we define the operation $\Psi_{\cup, \Lambda}$ to be:
    \begin{equation}
        \Psi_{\cup, \Lambda}: C^{1}(\underline{X}/\underline{G},\Zn{2})\times \cdots \times C^{1}(\underline{X}/\underline{G},\Zn{2})\to \Zn{2}, \quad  \Psi_{\cup, \Lambda}(\underline{x}_1,\dots, \underline{x}_\Lambda)=\int_\Lambda \underline{x}_1\cup\cdots\cup\underline{x}_\Lambda
    \end{equation}
   where $\int_\Lambda$ is the integral of dimension $\Lambda$. 
\end{definition}

And we have a similar result to that in the tensor product case:

\begin{lemma} [Sufficient condition for cohomology invariant on balanced product codes] \label{lemma:sufficient_cup_balanced}
Under assumptions of \cref{lemma:sufficient_cup}, with an addition of pre-orientation being preserved under the group action, the operation $\Psi_{\cup, \Lambda}$ on a balanced product complex $C(\underline{X}/\underline{G},\Zn{2})$ (\cref{eq:balancedproductquantum}) is a cohomology invariant. 
\end{lemma}
\begin{proof}
Using the properties that modding by the group action preserves pre-orientation and integrals, the proof is analogous to that of \cref{lemma:sufficient_cup}.
\end{proof}

The goal of the rest of this paper is to explore the logical gates that one can obtain on quantum codes using operation $\Psi_{\cup,\Lambda}$ (which is a tensor product or a balanced product of classical codes). We will do this assuming that the sufficient conditions \cref{lemma:sufficient_cup} and \cref{lemma:sufficient_cup_balanced} (imposed on constituent classical codes) for $\Psi_{\cup,\Lambda}$ to be a cohomology invariant holds, and will also discuss the possible routes to obtain classical codes satisfying these conditions.

\begin{remark}
Given the definition of pre-orientation and the cup product, we also have the  \emph{necessary condition} for the operation $\Psi_{\cup,\Lambda}$ to be a cohomology invariant, namely:
    \begin{equation}
        \int_\Lambda \underline{z}_1 \cup \dots \cup \underline{\delta}  (\underline{a}_s) \cup \dots   \cup \underline{z}_\Lambda  =0 \ (\op{mod} 2),
    \end{equation}
which can be further expanded using the properties of the cup product on a tensor product (balanced product) complex.  This could be a good starting point for proving no-go theorems for fault-tolerant implementation of certain classes of logical gates in specific kinds of quantum codes. 
\end{remark}

\begin{remark}
    Finally, we would like to emphasize one could further generalize the approach for constructing a cohomology invariant on the product complex by considering general multilinear functions on $\Lambda$ arguments on constituent classical codes. For the classical code corresponding to the complex $C(X^{(i)}, \Zn{2})$, we would have
    $$ \psi^{(i)}_{\Lambda}: C^{p_1}(X^{(i)}, \Zn{2})  \times \dots \times C^{p_\Lambda}(X^{(i)}, \Zn{2}) \rightarrow \Zn{2}$$
    which only have nontrivial action when $\sum_j p_j = 1$ (and we assume it already embeds an integral).  Respective operation on the quantum complex would be defined as
    $$\Psi_{\Lambda}(\underline{x}_1,\dots, \underline{x}_\Lambda) = [\psi^{(1)}_\Lambda (x_1^{(1)}, \dots, x_1^{(\Lambda)}), \dots, \psi^{(\Lambda)}_\Lambda (x_\Lambda^{(1)}, \dots, x_\Lambda^{(\Lambda)})]$$
    where $\underline{x}_1 = [x_1^{(1)}, \dots, x_\Lambda^{(1)}]$. 
    The Leibniz rule will have the form $\sum_{j=1}^\Lambda \psi(x_1,\dots,\delta(x_j), \dots,x_\Lambda)=0$, and will similarly be a sufficient condition for $\Psi_\Lambda$ to be a cohomology invariant.
    
    One then could take $\psi_2 (x_1,x_2) = x_1 \cup x_2$, and assume some more general definition for~$\psi_{\Lambda \geq 3}$ as long as it produces $\Psi_\Lambda$ that is a cohomology invariant. 
\end{remark}

\subsection{Cup products in group algebra codes} \label{sec:group-algebra-codes}

As a first example, we will now apply our discussion to the (balanced products of) group algebra codes.  This entails bivariate bicycle codes~\cite{PhysRevA.88.012311,ibm_qldpc} as a special case.

Let $G$ be a finite Abelian\footnote{The results here also apply to non-Abelian groups if one is careful about left and right multiplication.} group, and denote by $A=\ftwo[G]$ the group algebra of~$G$ over~$\ftwo$.
A \emph{group algebra code} can be defined using any element $c\in \ftwo[G]$ via the complex
$$C(\ftwo[G],c): \ftwo[G] \stackrel{\delta}{\longrightarrow} \ftwo[G], \quad  \delta(g)=c\cdot g$$
for $g \in \ftwo[G]$, which arises from multiplication with $c$ and where both terms have  $G$ for the basis.

We can consider balanced products of group algebra codes which will be associated with quantum codes. If we have a family of elements $c^{(i)}\in \ftwo[G]$, we can consider the balanced product of $\Lambda$ group algebra codes:
\begin{equation}
    C(\ftwo[G],\underline{c})=C(\ftwo[G],c^{(1)})\otimes_G\cdots \otimes_G C(\ftwo[G],c^{(\Lambda)}).
\end{equation}
If each constituent chain complex has a cup product and an integral that is additionally preserved by the group action, then $C(\ftwo[G],\underline{c})$ is equipped with a cup product and an integral following the discussion in the previous subsection.

Below, we discuss how to satisfy the sufficient conditions for $\Psi_{\cup, \Lambda}$ being a cohomology invariant on products of $\Lambda$ group algebra codes for separately for $\Lambda = 2$ and $\Lambda \geq 3$ since the sufficient conditions are of different strengths for these cases.

\subsubsection{Product of $\Lambda = 2$ group algebra codes}

Here, we ask how to construct codes that admit pre-orientation preserved by the group action, and the integral, and also satisfy the sufficient condition, namely, the Leibniz condition (\cref{eq:lambda2}) for $\Lambda = 2$. 

We give one possible way to equip a group algebra code given by $C(\ftwo[G],c)$ with these properties. Assume that the element $c$ admits a decomposition
$$c=c_{\op{in}}+c_{\op{out}} + c_{\op{free}}$$
such that $c_{\op{in}}$, $c_{\op{out}}$ and $c_{\op{free}}$ have disjoint support. Then we can split the boundary map into $\delta = \din + \dout + \delta_{\op{free}}$ via
$$\din(g)=c_{\op{in}}\cdot g\text{, }\dout(g)=c_{\op{out}}\cdot g, \text{ and }\delta_{\op{free}}(g)=c_{\op{free}}\cdot g. $$  

Below we provide a sufficient (but not necessary) condition for the integrated Leibniz rule for $\Lambda = 2$ to hold:

\begin{proposition}\label{prop:orientationsforgroupalgebracodes}
 Consider a classical code given by complex $C(\ftwo[G],c)$ with a boundary map as discussed above and an integral. Denote by $\op{inv}:\ftwo[G]\to\ftwo[G]$ the involution defined by $\op{inv}(g)=g^{-1}.$  
If \begin{itemize}
    \item [(1)] $|c_{\op{in}}| = 1$,
    \item [(2)] $c_{\op{in}}=\op{inv}(c_{\op{out}})$,
    \item [(3)] $c_{\op{free}}= \op{inv}( c_{\op{free}})$, and
    \item [(4)] The group $G$ acts freely on the basis of $\ftwo[G]$.
\end{itemize}
Then 
\begin{itemize}
    \item [(a)]  The group action preserves the splitting of the boundary map, namely 
    \begin{equation}
    \begin{split}
          &\delta _{\op{in}} (g \cdot h) = g \cdot \delta_{\op{in}} (h);\\
          &\delta _{\op{out}} (g \cdot h) = g \cdot \delta_{\op{out}}  (h);\\
          &\delta _{\op{free}} (g \cdot h) = g \cdot \delta_{\op{free}} (h);
    \end{split}
    \end{equation}
    \item [(b)] The splitting of the boundary map defines a pre-orientation that in addition, satisfies the non-overlapping bits condition. Thus, the chain complex $C(\ftwo[G],c)$ admits an associative cup product; 
    \item [(c)] Integrated Leibniz rule for $\Lambda = 2$ is fulfilled on the basis of \Cref{eq:lambda2}  being satisfied.
\end{itemize}
\end{proposition}
\begin{proof}

Because $c_{\op{in}}$, $c_{\op{out}}$ and $c_{\op{free}}$ have disjoint support and the group action is free on the basis of $\ftwo[G]$ (which we choose to be $G$), this directly yields (a).

Next, to satisfy, (b), we need to have $\din (g) \cap \din (h)  = \varnothing$ and $\dout (g) \cap \dout (h)  = \varnothing$ for $g \neq h \in C^0(\ftwo[G],c)$. Indeed, $\din (g) + \din(h) = c_{\op{in}} g +c_{\op{in}} h = c_{\op{in}} (g + h) \neq 0 $ unless $g = h$, and similarly for $c_{\op{out}}$.

The first condition from \Cref{eq:lambda2}, namely that for any $g \in G$ $\left | \din(g)\right | +\left | \dout(g)\right |   = 0 \ (  \op{mod} 2)$, it true because $|c_{\op{in}}| = |c_{\op{out}}|=1$.
Let us now check the second condition, namely, $\forall g \neq h \in G$
 \begin{equation*}
    \left | \dout(g) \cap \delta_{\op{free}}(h)\right | +\left | \delta_{\op{free}}(g) \cap \din(h) \right |   = 0 \ (  \op{mod} 2). 
 \end{equation*}
Consider \begin{align*}
    \left | \dout(g) \cap \delta_{\op{free}}(h)\right | &= \left | c_{\op{out}}g  \cap c_{\op{free}}h\right | = \left | \op{inv}(c_{\op{out}}) g ^{-1}   \cap \op{inv}(c_{\op{free}}) h^{-1} \right | = \left | c_{\op{in}} g ^{-1}   \cap c_{\op{free}}  h^{-1} \right |
    \\
    &=\left | c_{\op{in}} h   \cap c_{\op{free}}  g \right | = \left | \delta_{\op{free}}(g) \cap \din(h) \right |,
\end{align*}   
which completes the proof.

\end{proof}

For $\Lambda=2,$ by \cite[Appendix A]{eberhardt2024logicaloperatorsfoldtransversalgates} the balanced product complex $C(\ftwo[G],\underline{c})$ can be identified with

\[\begin{tikzcd}[column sep=40pt]
	{\ftwo[G]} & {\ftwo[G]^2} & {\ftwo[G]}
	\arrow["\left(\begin{smallmatrix} c^{(1)} \\ c^{(2)} \end{smallmatrix}\right)", from=1-1, to=1-2]
	\arrow["{(c^{(2)},c^{(1)})}", from=1-2, to=1-3]
\end{tikzcd}\]
and is an example of a two-block group algebra code, see \cite{linQuantumTwoblockGroup2023}. 
If moreover $G=\Zn{\ell}\times \Zn{m}$, then the code is called a \emph{bivariate bicycle code}, see \cite{PhysRevA.88.012311,ibm_qldpc}.

Let us compute an explicit formula for the cup product of two $1$-cochains. The space of $1$-cochains can be described as follows
\begin{align*}C^1(\ftwo[G],\underline{c})&= C^0(\ftwo[G],c^{(1)})\otimes_G C^1(\ftwo[G],c^{(2)})\oplus C^1(\ftwo[G],c^{(1)})\otimes_G C^0(\ftwo[G],c^{(2)})\\&=(\ftwo[G]\otimes_{G}\ftwo[G])_v\oplus (\ftwo[G]\otimes_{G}\ftwo[G])_h\\
&\cong\ftwo[G]_v\oplus \ftwo[G]_h.
\end{align*}
Here, the last isomorphism comes from the diagonal group action and the multiplication map $\ftwo[G]\otimes_{G}\ftwo[G]\to \ftwo[G]$, $[p,q]_G\mapsto pq$. 
In the following we translate back and forth with respect to this isomorphism. 
Moreover, we label the two copies of $\ftwo[G]$ by indices $v,h$ to differentiate them from each other. 
Now, let $p_h\in \ftwo[G]_h$ and $q_v\in \ftwo[G]_v$ be the basis vectors, where $p,q \in G$. 
We write $q_v=[1,q]_G$ and $p_h=[p,1]_G$ where $1$ is a check (an element of $C^0$) and $p$ and $q$ are bits (elements of $C^1$). 
Then by \Cref{eqn:cupproductbalancedproduct} the cup product can be computed as follows
\begin{align*}
    q_v\cup p_h&=[1,q]_G\cup [p,1]_G=\sum_{g\in G}[1\cup pg^{-1}, q\cup g]
=\sum_{g\in G}\delta_{pg^{-1}\in c^{(1)}_{\op{out}}}\,\delta_{q\in gc^{(2)}_{\op{in}}}[ pg^{-1}, q] \\
&=\sum_{h\in G}\delta_{h\in p^{-1}c^{(1)}_{\op{out}}}\,\delta_{h\in q^{-1}c^{(2)}_{\op{in}}}phq\\
&=\sum_{h\in \left (p^{-1}c^{(1)}_{\op{out}}\cap q^{-1}c^{(2)}_{\op{in}} \right)}phq.
\end{align*}

\begin{example}
    There are indeed interesting bivariate bicycle codes that admit an orientation and hence a cup product. A quick computer search gives a bivariate bicycle code with parameters $[[144, 8, 12]]$ fulfilling all the assumptions of \Cref{prop:orientationsforgroupalgebracodes} where 
    \begin{align*}
        &c^{(1)}_{\op{in}}= x^3y^2, \ c^{(1)}_{\op{out}}=x^{-3}y^{-2} , \ c^{(1)}_{\op{rest}} = x^2y+ x^{-2}y^{-1}
        \\
        &c^{(2)}_{\op{in}}=  x,  \ c^{(2)}_{\op{out}}=x^{-1} , \ c^{(2)}_{\op{rest}} = xy + x^{-1}y^{-1}
    \end{align*}
    and $G=\Zn{6}\times \Zn{12}$ with generators $x\in \Zn{6}$ and $y\in \Zn{12}$. Note that this is not a unique splitting of $c^{(i)}=c^{(i)}_{\op{in}}+c^{(i)}_{\op{out}} + c^{(i)}_{\op{free}}$, and different choices would subsequently lead to a different logical action. 
\end{example}

\subsubsection{Product of $\Lambda \geq 3$ group algebra codes}

For a classical group algebra code to satisfy the integrated Leibniz rule for $\Lambda \geq 3$, one could search for a code with a pre-orientation that satisfies \Cref{eq:lambda3}.
For now, we conjecture that such classical codes exist, and show the expression for the cup product on a balanced product of them.

For $\Lambda=3$,  the balanced product $C(\ftwo[G],\underline{c})$ yields the complex

\newcommand{\matrixAx}{\left(\begin{smallmatrix}c^{(h)} & c^{(v)} & c^{(d)}\end{smallmatrix}\right)}
\newcommand{\matrixBx}{\left(\begin{smallmatrix}c^{(d)} & 0 & c^{(v)} \\ 0 & c^{(d)} & c^{(h)} \\ c^{(h)} & c^{(v)} & 0\end{smallmatrix}\right)}
\newcommand{\matrixCx}{\left(\begin{smallmatrix}c^{(v)} \\ c^{(h)} \\ c^{(d)}\end{smallmatrix}\right)}

\begin{equation*}
    \begin{tikzcd}[ column sep =60pt]
	\ftwo[G] & { \substack{\ \ \ftwo[G]_v \\ \oplus \ftwo[G]_h \\ \oplus \ftwo[G]_d } } & { \substack{\ \ \ftwo[G]_{vd} \\ \oplus \ftwo[G]_{hd} \\ \oplus \ftwo[G]_{vh} } } & \ftwo[G],
	\arrow["\matrixCx", from=1-1, to=1-2]
	\arrow["\matrixBx", from=1-2, to=1-3]
	\arrow["\matrixAx", from=1-3, to=1-4]
\end{tikzcd}
\end{equation*}
Assuming an Abelian group $G$, we have $[p,q,r]_G=[p\cdot g^{-1} ,g \cdot q,r]_G = [p , h^{-1}\cdot q,r \cdot h]_G$, where $p,q,r \in G$.  Then, the multiplication map can be written as $\ftwo[G]\otimes_{G}\ftwo[G]\otimes_{G}\ftwo[G] \to \ftwo[G]$, $[p,q,r]_G\mapsto pqr$.
Let $p_{v} = [p,1,1]_G \in \ftwo[G]_{v}$, $q_{h} = [1,q,1]_G \in \ftwo[G]_{h}$ and $r_d = [1,1,r]_G \in \ftwo[G]_{d}$ be the basis vectors. Here, all 1 elements are checks (elements of $C^0$  of classical codes) and $p,q,r$ are bits (elements of $C^1$ of classical codes). \hl{Using Eq.~(\ref{eq:triple-cup}) the formula for the cup product can be written as}
\begin{align*}
     p_v\cup q_h \cup r_d&=\sum_{g_v,h_v\in G}\sum_{g_h,h_h\in G}\sum_{g_d,h_d\in G} [pg_v^{-1},g_vh_v^{-1},h_v]\cup [g_h^{-1},g_hh_h^{-1}q,h_h] \cup [g_d^{-1},g_dh_d^{-1},h_dr]\\
     &=\sum_{g_v,h_v\in G}\sum_{g_h,h_h\in G}\sum_{g_d,h_d\in G} [pg_v^{-1}\cup g_h^{-1}\cup g_d^{-1},g_vh_v^{-1}\cup g_hh_h^{-1}q\cup  g_dh_d^{-1},h_v\cup h_h \cup h_dr]
\end{align*}
\hl{consider the first and the third tensor factors. From decomposition $c_{\alpha}=c_{\alpha,\op{in}}+c_{\alpha,\op{out}} + c_{\alpha,\op{free}}$ for each $\alpha = h,v,d$, we see that the non-overlapping bits condition is always satisfied (see also the proof of Prop.~\ref{prop:orientationsforgroupalgebracodes}, and therefore, the cup product is associative. 

Using associativity, we see that the expression above is nonzero only if $g_h = g_d$ and $h_h = h_v$. Then, the first tensor factor becomes $pg_v^{-1}\cup g_d^{-1}\cup g_d^{-1} = pg_v^{-1}\cup g_d^{-1}$ and the last tensor factor, similarly, becomes $h_v \cup h_dr$. Upon further simplifying the expression, we obtain}
\begin{align*}
    p_v\cup q_h \cup r_d
     &=\sum_{g_v,h_v\in G}\sum_{g_d,h_d\in G} (g_v,h_d)\cdot[p\cup g_vg_d^{-1},h_v^{-1}h_d\cup g_v^{-1}g_dh_v^{-1}h_dq\cup  g_dg_v^{-1}, h_vh_d^{-1}\cup r]\\
     &=\sum_{g',h'\in G}\sum_{g,h\in G} (g',h')\cdot[p\cup g,h\cup g^{-1}hq\cup  g^{-1}, h^{-1}\cup r]\\
     &=\sum_{g,h\in G} [p\cup g,h\cup g^{-1}hq\cup  g^{-1}, h^{-1}\cup r]_{G\times G}
\end{align*}
\hl{We can now explicitly evaluate the cup products in this expression. We obtain $p\cup g = p$ if $p \in g c_{in}^{(v)}$ and zero otherwise; $h^{-1} \cup r = r$ if $r h \in c^{(d)}_{out}$ and zero otherwise; and, finally $h \cup g^{-1} h q \cup g^{-1} = g^{-1}h q$ if $hq \in c^{(h)}_{in}$ and $q \in g c^{(h)}_{out}$ and zero otherwise. Collecting everything together and simplifying the multiplication map, we obtain} 
\begin{align*}
    p_v\cup q_h \cup r_d =\sum_{\tiny{\substack{
     g\in \left (p^{-1}c^{(v)}_{in}\cap q^{-1}c^{(h)}_{out} \right)\\
     h\in \left( q^{-1}c^{(h)}_{in}\cap r^{-1}c^{(d)}_{out} \right)}
     }} pgqhr.
\end{align*}

% \jne{here is the redone computation:}
% \begin{align*}
%      p_v\cup q_h \cup r_d& 
%      =\sum_{g,h,g',h'\in G}[pg^{-1}\cup 1 \cup g'^{-1},gh^{-1}\cup q\cup g'h'^{-1},h\cup 1 \cup h'r]_{G\times G}\\
%      &=\sum_{g,g',h'\in G}[pg^{-1}\cup 1 \cup g'^{-1},g\cup q\cup g'h'^{-1},1 \cup h'r]_{G\times G}\\
%      &=\sum_{\substack{g,g',h'\in G\\ pg^{-1}\in c^{(v)}_{\op{in}}\\q\in c^{(h)}_{\op{out}}\cdot g}}[pg^{-1} \cup g'^{-1},q\cup g'h'^{-1},1 \cup h'r]_{G\times G}\\
%      &=\sum_{\substack{g\in G\\ pg^{-1}\in c^{(v)}_{\op{in}}\\q\in c^{(h)}_{\op{out}}\cdot g}}\sum_{\substack{g',h'\in G\\pg^{-1}\in c^{(v)}_{\op{in}}\cdot g'^{-1} \\q\in c^{(h)}_{\op{in}}\cdot g'h'^{-1}\\h'r\in c^{(d)}_{\op{out}}}}[pg^{-1},q,h'r]_{G\times G}\\
%      &=\sum_{\substack{g\in G\\ g^{-1}\in (p^{-1}c^{(v)}_{\op{in}}\cap q^{-1}\in c^{(h)}_{\op{out}})}}\sum_{\substack{g',h'\in G\\g'\in gp^{-1}\cdot c^{(v)}_{\op{in}} \\ h'g'^{-1} \in q^{-1}c^{(h)}_{\op{in}}\\h'\in c^{(d)}_{\op{out}}r^{-1}}}[pg^{-1},q,h'r]_{G\times G}.
% \end{align*}

\begin{remark}
    \hl{After the current paper appeared online, several independent  works \cite{menon2025magictricyclesefficientmagic,ryan,2508.08191,li2025transversaldimensionjumpproduct} found interesting  non-trivial examples of group algebra codes that admit non-trivial non-Clifford gates using the formalism presented here.}
\end{remark}

% \jne{
% Here we could add a lemma that says:
% Assume that there is an involution $\sigma$ of $G$ such that $\sigma(g\cdot -)=g^{-1}\sigma(-)$.
% Moreover, choose $c_{\op{in}}$ such that $|c_{\op{in}}|\leq 2.$ Then set $c_{\op{out}}=\sigma(c_{\op{in}})$ and let $c_{\op{free}}=0$. Then we obtain a cup product for $\Lambda=3.$
% }

\subsection{Sipser--Spielman codes}  \label{sec:sipser-spielman}
We can also consider tensor and balanced products of Sipser--Spielman codes~\cite{sipser1996expander}. 
Sipser--Spielman codes are constructed by adding a local code to a graph code, where the graph is expanding\footnote{This is a special case of a code construction by Tanner~\cite{tanner1981recursive}.}.
This code construction yields good classical LDPC codes using, for example, the Ramanujan graphs constructed by Margulis~\cite{margulis}.
Balanced products of Sipser--Spielman codes were constructed in \cite{balancedproductcodes,qldpc_review} and conjectured to yield good qLDPC codes. This was subsequently proven in Ref.~\cite{PK2022}.

Let us show several particular examples of constructing Sipser-Spielman codes with integrated cup products in this subsection. Beyond these examples, it should be possible to prove the existence of good classical LDPC codes with nontrivial integrated cup products that satisfy the sufficient conditions in \cref{prop:integrated-leibniz-conditions} by mapping the constraint satisfaction onto a problem of the colorability of the Tanner graph of the code. Moreover, since these conditions are sufficient and not just necessary, there should be even more flexibility for achieving this.

The construction starts with an $s$-regular graph $X$, which we interpret as a 1D simplicial complex here. 
For a vertex $v\in X_0,$ denote by $X_v\subset X_1$ the set of edges which are incident to $v$ in $X.$
The cochain complex associated to $X$ is denoted by 
$$C^0(X,\Zn{2})\xrightarrow{\delta_X} C^1(X,\Zn{2})$$
and the coboundary operator maps a vertex $v\in X_0$ to the sum of incident edges
$$\delta_X(v)=\sum_{e\in X_v} e.$$
Choosing an orientation of each edge, we obtain a cup product on this complex, see \Cref{sec:cupproduct}. 
Moreover, if we set $\int\!e=1$ for each edge, then $\int \delta_X(v)=s$. This is a $1$-integral on $C(X,\Zn{2})$ if and only if $s$ is even (cf.~\Cref{sec:integrals}), which we assume from now.

We now explain how to enhance this by a local system, following ideas of Sipser--Spielmann, where the interpretation as local systems on graphs is due to \cite{meshulam2018graph}, see also \cite[Section~III.A]{balancedproductcodes}.
Fix a based complex, which we refer to as a local system or local code,  
$$C^0(L,\Zn{2}) \xrightarrow{\delta_L}  C^1(L,\Zn{2}) $$
with bases $L_0$ in $C^0(L,\Zn{2})$ and $L_1$ in $C^1(L,\Zn{2})$, such that $|L_1|=s.$  Assume that for each vertex $v$ we are given a bijection $\phi_v:L_1\to X_v$. By abuse of notation, we also denote the obvious induced map by $\phi_v: C^1(L,\Zn{2})\to C^1(X,\Zn{2})$.
Using this, we define a new complex,
$$C^0(X,\Zn{2})\otimes L^0\xrightarrow{\delta} C^1(X,\Zn{2})$$
with coboundary
$$\delta(v\otimes c)=\phi_v(\delta_L(c)).$$
The complex has a natural basis $X^L=X_0\times L_0\uplus X_1,$ and we abbreviate the associated complex by $C(X^L,\Zn{2}).$

\subsubsection{The case $\Lambda=2$}

Let us show one example of how to define the pre-orientation on the complex $C(X^L,\Zn{2})$ that satisfies conditions of \cref{prop:integrated-leibniz-conditions} for $\Lambda=2$. 

First, we notice that the complex $C(X,\Zn{2})$ on the graph $X$ always has a pre-orientation induced by an arbitrary choice of orientation on the graph. Assign the edges that come into the check $x$ to be in $\din(x)$, and the outgoing edges to be in $\dout(x)$, which gives $\delta (x) = \din(x) \uplus \dout (x)$. For example, if $X=\op{Cay}^{\ell/r}(G)$ is the left/right Cayley graph of group $G$, one can always choose one-half of generators to correspond to $\din$ for every bit and the rest to correspond to $\dout$ (note $s$ is always assumed to be even).

\begin{proposition}\label{thm:constructionorientationsipserspielman}
   We assume that there exists an orientation on the graph $X$ and a pre-orientation on the complex $C(X^L,\Zn{2})$  such that in addition, the following assumptions hold:
    \begin{enumerate}
        \item There is a partition $L_1=L_1^{\op{in}}\uplus L_1^{\op{out}}$ such that $\phi_v(L_1^{\op{in}})$ and $\phi_v(L_1^{\op{out}})$ are contained in the incoming and outgoing edges, correspondingly,  with respect to the orientation on $X$.
        \item The pre-orientation $\delta_{L,\op{in}}(c)=\delta_L(c)\cap L_1^{\op{in}}$, $\delta_{L,\op{out}}(c)=\delta_L(c)\cap L_1^{\op{out}}$ on the local code satisfies the conditions of \cref{prop:integrated-leibniz-conditions}  for $ \Lambda = 2$ on the local code.
    \end{enumerate}
    Then the pre-orientation defined by
    \begin{equation}
        \begin{split}
            &\din(v \otimes c)=\phi_v(\delta_{L}(c)\cap L_1^{\op{in}}),
            \\
            &\dout(v \otimes c)=\phi_v(\delta_{L}(c)\cap L_1^{\op{out}}).
        \end{split}
    \end{equation}
    satisfies the conditions of \cref{prop:integrated-leibniz-conditions} for the complex $C(X^L,\Zn{2})$ for $\Lambda = 2$.
\end{proposition}
\begin{proof}
    Follows directly from verifying the conditions of \cref{prop:integrated-leibniz-conditions} given the assumptions.
\end{proof}

\begin{example}\label{ex:cayleygraphnicegenerators}
Let $G$ be a finite group and choose its generating set $S$, which we decompose as $S=T\uplus T^{-1}$. Denote by $X=\op{Cay}^{\ell}(G,S)$ the (left) Cayley graph. 
The graph has a natural orientation via the decomposition $S=T\uplus T^{-1}$ where edges coming from an element in $T$ are `outgoing' and edges coming from an element in $T^{-1}$ are `incoming'. We assign no edges to be `free'. 
For example, if for each check $c$ in the local code $L$ we have $t\in \delta_L(c)$ implies $t^{-1}\in \delta_L(c)$ then the conditions of \Cref{thm:constructionorientationsipserspielman} are fulfilled and there exists a $\Lambda = 2$-fold cup product. 

Now assume that we are given such local codes $L$ and $L'$ and orientations for both the left and right Cayley graphs $\op{Cay}^{\ell}(G,S)$ and $\op{Cay}^{r}(G,S)$. 
Then the balanced product
$$C(\op{Cay}^{\ell}(G,S)^{L}\times_G\op{Cay}^{r}(G,S)^{L'},\Zn{2})$$
fulfills the conditions of \cref{lemma:sufficient_cup_balanced}.
\end{example}

\subsubsection{The case $\Lambda\geq 2$} 

Let us show an example of a way to fulfill the integrated Leibniz rule for $\Lambda\geq 2$ on a Sipser--Spielman code. The idea is based on \cref{prop:hidden-graph}: if one can find a graph structure ``inside'' a complex, one can equip this complex with a cup product that satisfies the integrated Leibniz rule, for all $\Lambda \geq 2$.

As in \Cref{ex:cayleygraphnicegenerators}, we start with a group $G$ and its generating set $S$ that we split as $S=T\uplus T^{-1}$, and as well as a local code $C(L,\Zn{2})$ and bijection $\phi: L_1\to S$ between the bits of the local code and the set of edges (labeled by $S$) incident to each group element $g\in G$ in the Cayley graph $\op{Cay}(G,S) \equiv X.$

Now, we fix an element $\hat{t}\in T$ and make the following assumption. We assume that there exists a check $\hat{c}\in L_0$ in the local code such that $\phi(\delta_L(\hat c))=\{\hat{t},\hat{t}^{-1}\}$ and such that for all other checks  $c\neq \hat c\in L_0$ we have $\phi(\delta_L(c))\cap\{\hat{t},\hat{t}^{-1}\}=\varnothing$.
Then, we define a pre-orientation on the local code $L$ via
\begin{equation} \label{eq:pre-orientation-SS-3}
\begin{split}
    \delta_{L,\op{in}}(\hat{c}) &= \phi^{-1} \left (\{\hat{t}\} \right ),\\
    \delta_{L,\op{out}}(\hat{c}) &= \phi^{-1}  \left(\{\hat{t}^{-1}\} \right ),\\
    \delta_{L,\op{free}}(\hat{c}) &= \delta_L(\hat{c})\setminus \phi^{-1}  \left(\{\hat{t},\hat{t}^{-1}\} \right )\text{ and}\\
    \delta_L(c) &\equiv \delta_{L,\op{free}}(c)\text{ for all }c\neq \hat{c}.
    \end{split}
\end{equation}

This pre-orientation induces a pre-orientation on the Sipser-Spielman code according to: 
\begin{equation*}
        \begin{split}
            \din(v \otimes c)&=\phi_v( \delta_{L,\op{in}}(c)),
            \\
            \dout(v \otimes c)&=\phi_v( \delta_{L,\op{out}}(c)), \text{ and}
             \\
            \delta_{\op{free}}(v \otimes c)&=\phi_v( \delta_{L,\op{free}}(c)),
        \end{split}
    \end{equation*}

This is similar to \cref{prop:hidden-graph}, except here we explicitly formulate the ``hidden'' graph to be a cycle that is connected by elements $t$ through the check $\hat c$ on each vertex.

\begin{proposition} \label{sslarger3}
    The Sipser-Spielman code defined above
    satisfies conditions in \cref{prop:integrated-leibniz-conditions}  for all $\Lambda\geq 2.$
\end{proposition}
\begin{proof}
    Since the condition at $\Lambda \geq 3$ implies that for smaller $\Lambda$, we only need to check the ``non-overlapping bits'' condition and \cref{eq:lambda31}. The ``non-overlapping bits'' condition holds because of the partitioning chosen in \cref{eq:pre-orientation-SS-3}. Namely, $\din(v \otimes c) \cap \din (v' \otimes c')$ is empty if $v \neq v'$ because the incoming edges are always associated with $t$ and outgoing with $t^{-1}$, and similarly for $\dout$.  If $v = v'$, from \cref{eq:pre-orientation-SS-3} we have that only one check has nontrivial $\delta_{\op{in}/\op{out}}$. 

     To check the conditions in \cref{eq:lambda31}, we decompose $${\delta (v \otimes c) = \din (v \otimes c) + \dout (v \otimes c) + \delta_{\op{free}} (v \otimes c)}$$ and consider all possibilities for $(v_1 \otimes c_1, v_2 \otimes c_2, v_3 \otimes c_3)$, i.e.  $v_1 = v_2 = v_3$,   $v_1 = v_2 \neq v_3$, $\{v_1 \neq v_2 \neq v_3, v_1 \neq v_3\}$, etc. We find that all four conditions hold directly from \cref{eq:pre-orientation-SS-3}.
\end{proof}

Let us show that the check $\hat c$ can be chosen such that the integrated cup product gives a nontrivial cohomology operation. 

\begin{proposition} \label{prop:nontrivial-action-SS}
    Under the assumptions of \cref{sslarger3}, assume that the local check $\hat{c}$ that checks all bits, i.e.  $\delta(\hat{c})=L_1$.
    Denote $\underline{c}=\sum_{v \in C^0(X,\Zn{2})}v\otimes \hat{c} \ \in \ C^0(X^L,\Zn{2})$.  
    For any fixed element $v\in C^0(X,\Zn{2})\equiv G$ denote by $e_v$ and $e'_v$ the incident edges labeled by~$\hat t$ and~$\hat t^{-1}$ (which can be uniquely determined for each vertex $v$ by assumption). 
    Then the following statements hold.
    \begin{enumerate}
        \item[(1)] $e_v \cup \underline{c}=e_v$, $\underline{c}\cup e_v'=e_v',$ and  $\underline{c}\cup \underline{c}=\underline{c}$, and all cup products not containing these are 0;
        \item[(2)] $\underline{c}\in Z^0(X^L,\Zn{2})=H^0(X^L,\Zn{2})$;
        \item[(3)] $[e_v]$ and $[e'_v]$ are non-trivial in $H^1(X^L,\Zn{2})$.
        \item[(4)] As a consequence, $\Lambda$-fold integrated cup product is non-trivial, for all $\Lambda\geq 1.$
    \end{enumerate}
\end{proposition}
\begin{proof}
    To show (1), we note that $e_v\cup \left ( v' \otimes \hat c\right) = e_v \cap \din \left ( v' \otimes \hat c\right) = e_v \cap e_{v'} = e_v \delta_{v = v'}$, and similarly one can show that $ \left ( v' \otimes \hat c\right) \cup e'_v = e'_v \delta_{v = v'}$. The third relation holds straightforwardly. Any other primitive relation is zero because the only boundary operators involved in the cup product are $\din$ and $\dout$, which are in correspondence with $e_v$ and $e_v'$.  

    To show statement (2), we write $$\delta(\underline c) = \sum_v  \left ( \phi_v(\delta_{L,\op{in}}(\hat c)+   \phi_v(\delta_{L,\op{out}}(\hat c) \right) =  \sum_v (e_v + e_v') = \sum_v e_v +\sum _w e_w = 0. $$ In the first equality, we used $\delta_{L,\op{free}}(\hat c) = 0$.  We then used that each incoming edge for a given vertex $v$ is also an outgoing edge for some vertex $w$, similarly to a graph. 

    Next, statement (3) follows directly from the fact that $e_v$ (and $e_v'$) are defined uniquely for each $v$, so they only appear in $\din$ and $\dout$ boundaries and only of elements of the form  $w \otimes \hat c$; restricting to this, the problem is reduced to a cyclic graph, where we know $[e_v]$ ($[e_v']$) is nontrivial. 

   Finally, consider $\int_1 a_1 \cup \dots \cup a_\Lambda$, where $a_i \in C^{p_i}(X^L, \Zn{2})$ $p_i = {0,1}$ and $\sum p_i = 1$ (otherwise, the integral will give 0).  Because of (1), we see that combinations of $\underline{c}$ and $[e_v]$ ($[e_v']$) will give a nontrivial outcome, thus, showing (4). 
    \end{proof}

\section{Gates from Cup Products}
\label{sec:highercontrolz}

In this section, we finally discuss quantum logical gates obtained from cohomology invariants that we discussed up till now. 

To obtain interesting logical gates, we set up a framework for constructing logical gates between $\Lambda$ copies of a quantum code~$\mathcal C$. 
The idea behind this framework is that if the code~$\mathcal C$ admits a $\Lambda$-fold cup product (or, more generally, a cohomological operation $\Psi_\Lambda$ that is multilinear), then $\Lambda$ copies admit a logical gate constructed using the cup product formalism. We call such gates \emph{copy-cup gates}. 
We show that copy-cup gates can implement logical gates that are arbitrarily high in the Clifford hierarchy, depending on the number $\Lambda$ of copies of the quantum code.
Finally, we construct several quantum code families with different asymptotic parameters that can achieve gates at any desired level of the Clifford hierarchy and discuss a few particular examples from them.

\subsection{Logical operations between copies of a quantum code}\label{sec:logicalsoncodeblocks}

Following the discussion in \cref{sec:operations-gates},  given a quantum code corresponding to a chain complex $C(\underline{X}, \Zn{2})$ and a cohomology operation $\Psi$ on this code that is also defined at the level of 1-cochains, we can derive a unitary circuit implementing a logical gate corresponding to this operation. When $\Psi$ is constructed as a sum (integral) of elementary operations that involve a constant number (say, $\Lambda$) 1-cochains, such as cup product, the resulting circuit will be a constant-depth one\footnote{The circuit is, in fact, $\Lambda$-local. It is also geometrically local if the code is defined through a simplicial complex.}. Let us describe a general framework of how to determine a logical gate corresponding to a cohomology operation $\Psi$ and derive the respective circuit. 

Let $C(\underline{X},\Zn{2})$ be a based complex and consider an operation
$$\Psi_\Lambda: C^{1}(\underline{X},\Zn{2})\times \dots \times C^{1}(\underline{X},\Zn{2})\to \Zn{2}$$
to be a multilinear map on $\Lambda$ copies of the complex, that is, $\Psi_\Lambda$ is linear in each component, and assume that $\Psi_\Lambda$ is also a cohomology operation, that is, $\Psi_\Lambda$ induces a well-defined map 
$$\Psi_\Lambda: H^{1}(\underline{X},\Zn{2})\times \dots \times H^{1}(\underline{X},\Zn{2})\to \Zn{2}.$$
For each copy $i \in \{1, \ldots, \Lambda \}$, the respective copy of the chain complex $C(\underline{X},\Zn{2})$ is put in correspondence with the quantum code $\mathcal C_i$ whose codespace lies in the Hilbert space $\mathcal{H}_i$. The underlying physical qubits of each code are associated with  $X_1$ (which is a basis of $C^{1}(\underline{X},\Zn{2})$). 

We denote the quantum code that arises from taking the union of $\Lambda$ copies of this quantum code as $\mathcal{C}$. In the following, by abuse of notation, we will use the same symbol for the codespace as for the code where it does not cause ambiguity. 
Its underlying physical qubits correspond to the disjoint union $\underline{X}_{1}\uplus \dots \uplus \underline{X}_{1}.$
We write the states in the associated Hilbert space  $\mathcal{H}$ as
$\ket{c_1,\dots,c_\Lambda}$ for $c_i\in C^{1}(\underline{X},\Zn{2})$ belonging to $i$-th copy of the code.

\begin{definition}\label{def:unitaryfromcohomologyoperation}
    We define the unitary associated with the operation $\Psi_\Lambda$ to be
    $$U_{\Psi, \Lambda}=\sum_{\substack{\{c_i\in C^{1}(\underline{X},\Zn{2}) \} \\ i \in [1,\Lambda]}}(-1)^{\Psi_\Lambda(c_1,\dots,c_\Lambda)}\ket{c_1,\dots,c_\Lambda}\bra{c_1,\dots,c_\Lambda}.$$
\end{definition}

\noindent For a set of $\Lambda$ physical qubits $x_i\in \underline{X}_{1}$ with $i \in [1, \Lambda]$, one from each quantum code copy, we will use the multi-controlled $\op{Z}$-gate
$$C^{\Lambda-1}\!Z_{x_1,\dots,x_\Lambda}$$
where our notation assumes $C^{0}Z = C$, $C^{1}Z = CZ$, $C^{2}Z = CCZ$ etc. Namely, this gate acts by multiplication by $-1$ if all physical qubits are in the state $\ket{1}$ and trivial otherwise.

\begin{definition}\label{def:circuitfromcohomologyoperation}
    We define the circuit associated to $\Psi_\Lambda$  to be the product of multi-controlled $Z$-gates
    $$\gateZ_{\Psi,\Lambda}=\prod_{\substack{\{x_i \in \underline{X}_{1} \} \\ i \in [1, \Lambda]}}(C^{\Lambda-1}\!Z_{x_1,\dots,x_\Lambda})^{\Psi_\Lambda(x_1,\dots,x_\Lambda)}.$$
\end{definition}
\noindent We notice that the circuit $\gateZ_{\Psi,\Lambda}$ is simply the same as the unitary above, except it is expressed through the chosen basis of physical qubit states. In addition, it is a constant-depth circuit. The result above establishes that $U_{\Psi,\Lambda} \equiv \gateZ_{\Psi,\Lambda}$ and is a logical gate on the code $\mathcal{C}$.

\begin{theorem}[Cups \& Gates I]\label{thm:operationyieldsgate}
Assuming an operation $\Psi_\Lambda$, corresponding unitary $U_{\Psi,\Lambda}$ and the circuit $\gateZ_{\Psi,\Lambda}$ as discussed above:

    \begin{enumerate}
        \item [(a)] The unitary $U_{\Psi,\Lambda}$ preserves the codespace $\mathcal{C}\subset \bigotimes_{i = 1}^{\Lambda}\mathcal{H}_i$. Thus, it is a logical gate for the code $\mathcal{C}$. 
        \item[(b)] Moreover, it can be equivalently written as the circuit $\gateZ_{\Psi,\Lambda},$ i.e. $U_{\Psi,\Lambda}=\gateZ_{\Psi,\Lambda}.$
    \end{enumerate}
\end{theorem}
\begin{proof}
   The code space $\mathcal{C}$ has a basis, up to normalization, of the form
    $$\ket{[\gamma_1],\dots,[\gamma_\Lambda]}=\sum_{b_1\in B^{p_1}(X,\Zn{2})}\dots \sum_{b_\Lambda\in B^{p_\Lambda}(X,\Zn{2})}\ket{\gamma_1+b_1,\dots, \gamma_\Lambda+b_{p_\Lambda}}$$
    where the basis states are labeled by the cohomology classes $[\gamma_i]\in H^{1}(X,\Zn{2})$ that are represented by cocycles $\gamma_i\in Z^{1}(X,\Zn{2}).$ 
    The unitary $U_{\Psi,\Lambda}$ acts on each summand via $(-1)$ to the power of $$\Psi_\Lambda(\gamma_1+b_1,\dots,\gamma_\Lambda+b_\Lambda)=\Psi_\Lambda(\gamma_1,\dots,\gamma_\Lambda).$$
    Hence, $U_{\Psi,\Lambda}$ multiplies each basis vector of $\mathcal{C}$ by a scalar, and hence, preserves $\mathcal{C}$.

    Now write write $c_i=\sum \epsilon_{{x_i}}x_i\in C^{1}(X,\Zn{2}).$ Then, by definition, $U_{\Psi,\Lambda}$ acts on $\ket{c_1,\dots,c_\Lambda}$ by multiplying it by $(-1)$ to the power of 
    \begin{align*}
        \Psi_\Lambda(c_1,\dots,c_\Lambda)=
        \sum_{x_1\in X_{p_1}}\dots \sum_{x_d\in X_{p_\Lambda}} \prod_i \epsilon_{{x_i}}\Psi_\Lambda(x_1,\dots,x_\Lambda)
    \end{align*}
    which agrees with the action of the circuit $\gateZ_{\Psi,\Lambda}.$
\end{proof}

To determine the logical gate implemented by   $U_{\Psi,\Lambda}$, let us choose a basis $\Gamma_i$ of each  $H^{1}(\underline{X},\Zn{2})$ for each copy of the code $i \in \{ 1, \ldots, \Lambda\}$. 
This corresponds to choosing a logical basis in the copies of the associated quantum codes $\mathcal{C}_i$ and in the total code $\mathcal{C}$.

\begin{lemma}\label{thm:explicitdescriptionofgate}
    The logical action of the unitary $U_{\Psi,\Lambda}$ (or, equivalently, the circuit $\gate_{\Psi,\Lambda}$) is 
    $$\overline{\gateZ}_{\Psi,\Lambda}=\prod_{\substack{[\gamma_i]\in\Gamma_i \\ i \in [1, \Lambda]}}(\overline{\op{C}^{\Lambda-1}Z}_{[\gamma_1],\dots,[\gamma_\Lambda]})^{\Psi_\Lambda(\gamma_1,\dots,\gamma_\Lambda)}.$$
\end{lemma}
\begin{proof}
    This follows as in \Cref{thm:operationyieldsgate} using the multilinearity of $\Psi_\Lambda.$
\end{proof}

The following result establishes the relation between a multilinear map $\Psi_\Lambda$ and a gate at a given level of Clifford hierarchy:

\begin{theorem} \label{th:guaranteed-level}
Assume that the map   $\Psi_\Lambda: H^{1}(X,\Zn{2})\times \dots \times H^{1}(X,\Zn{2})\to \Zn{2}$ that from the discussion above is non-zero. 
Then the unitary $U_{\Psi,\Lambda}$ (or equivalently the circuit~$\gateZ_\Psi$) implements a logical action at the $\Lambda$-th level Clifford hierarchy.
\end{theorem}
\begin{proof} By assumption, there must exist subset of basis vectors $[\gamma_i]\in \Gamma_i$ such that $\Psi_\Lambda(\gamma_1,\dots, \gamma_\Lambda) = 1$. Hence the circuit $\overline{\gateZ}_{\Psi,\Lambda}$ contains a non-trivial $\overline{\gateZ}_{\Psi,\Lambda}$-gate which is in the $\Lambda$-th Clifford hierarchy. Note that no cancellations can occur which would reduce the level of the Clifford hierarchy.
\end{proof}

Finally, let us discuss how a higher-level gates give rise to \emph{addressable} gates at lower levels of the Clifford hierarchy. Recall that a representative of a logical operator $\overline{X} = \overline{X}_{\gamma_i}$ acts as 
$$\overline{X}_{\gamma_i} \ket{c_1, \dots,  c_\Lambda} = \ket{c_1, \dots, c_i + \gamma_i, \dots, c_\Lambda}.$$ 
Therefore, , the circuit $\gate_{\Psi,\Lambda}^{\gamma_i} \equiv \overline{X}_{\gamma_i} \gate_{\Psi,\Lambda}^\dagger 
 \overline{X}_{\gamma_i} \gate_{\Psi,\Lambda}$ can be used to implement a gate at a lower level of the Clifford hierarchy than $\Lambda$. 
We note that we only obtain a nontrivial gate for the logical states for which $[\gate_{\Psi,\Lambda}^\dagger, 
 \overline{X}_{\gamma_i}] \neq 0$. This process can also be further repeated to get diagonal gates lower in the hierarchy.

\subsection{The copy-cup logical gate}

Finally, we can determine the logical action corresponding to the cohomology operation~$\Psi_{\cup, \Lambda}$ that we constructed in \cref{sec:cup-non-topological}.  We will call this gate the \emph{copy-cup} gate\footnote{Not to be confused with \emph{copy cat} or \emph{coffee cup}.} because it acts on $\Lambda$ copies of the same quantum code $\mathcal{C} = \bigoplus_{i = 1}^\Lambda C(\underline{X},\Zn{2})$ and its action is built on $\Lambda$-fold cup product (and an integral).  

As a reminder, the function $\Psi_{\cup, \Lambda}$ defined in \cref{sec:cup-non-topological} is multilinear and descends to cohomology, because it is a combination of the cup product and the integral, which both possess these properties. 
Following the discussion in this section, we obtain the unitary associated with~$\Psi_{\cup, \Lambda}$, which can be written as:
$$U_{\cup, \Lambda} = \sum_{ \substack{\{c_1\in C^{1}(\underline{X},\Zn{2}) \} \\ i \in [1, \Lambda]} }(-1)^{\int\!c_1\cup\dots\cup c_\Lambda}\ket{c_1,\dots,c_\Lambda}\bra{c_1,\dots,c_\Lambda}.$$
We call it the \emph{copy-cup unitary}. 
Similarly as above, we construct  the associated circuit ${\gateZ_{\cup, \Lambda}=\gateZ_{\Psi_{\cup, \Lambda}}}$, which we call the \emph{copy-cup circuit}.

As a direct application of  \Cref{thm:operationyieldsgate}, we obtain our main theorem:
\begin{theorem}[Cups \& Gates II]

Given the operation $\Psi_{\cup, \Lambda}$,

    \begin{enumerate}
        \item [(a)] The unitary $U_{\cup, \Lambda}$ preserves the codespace $\mathcal{C}\subset \bigotimes_{i = 1}^{\Lambda}\mathcal{H}_i$. Thus, it is a logical gate for the code $\mathcal{C}$. 
        \item[(b)] Moreover, it can be realized  by the copy-cup circuit $\gateZ_{\cup, \Lambda}$ whose physical action lies at $\Lambda$-th level of the Clifford hierarchy.
    \end{enumerate}
\end{theorem}
\begin{proof}
    Is a direct consequence of \cref{thm:operationyieldsgate}. 
\end{proof}

The specific logical action implemented by the unitary $U_{\cup, \Lambda}$ (or, equivalently, the circuit~$\gateZ_{\cup, \Lambda}$), i.e. depends on the details of the input code and on details of the cup product that one defines on it through pre-orientation. For example, if the operation satisfies the conditions of \cref{th:guaranteed-level}, then the logical action is at the level $\Lambda$ of the Clifford hierarchy.

\subsection{Families of qLDPC codes ascending the Clifford hierarchy}

With the approach presented above, one can get creative in constructing families of qLDPC codes with fault-tolerant gates on any level of the Clifford hierarchy. In this subsection, we attempt to provide several examples of families that are equipped with such gates. Some of these examples are for codes corresponding to chain complexes defined on manifolds, and some are built on operations that we discussed in \cref{sec:group-algebra-codes} and \cref{sec:sipser-spielman}. 

The first family of codes are the toric codes, corresponding to $C(T^\Lambda \uplus ... \uplus T^\Lambda, \Zn{2})$. \hl{This example appears earlier in the literature, for example, in Ref.~\cite{Wang2023}.}

\begin{theorem}[Poly-distance family with gate at level $\Lambda$]\label{thm:codeblockstoricgates}
    Consider $\Lambda$ copies of a $\Lambda$-dimensional hypercubic tori $X=T^\Lambda\uplus \dots \uplus T^\Lambda$ with side lengths $L$.
    By placing qubits on the edges, we obtain a quantum code corresponding to the complex $C(T^\Lambda \uplus ... \uplus T^\Lambda, \Zn{2})$ that encodes $\Lambda^2$ logical qubits into $\Lambda^2L^\Lambda$ physical qubits with distance $L$. Then:
    \begin{enumerate}
        \item One can define the copy-cup unitary $U_{\cup, \Lambda}$ as shown above  that implements the logical operator
    $$\prod_{\pi\in S_\Lambda}\overline{\op{C}^{\Lambda-1}\!\op{Z}}_{\left[\gamma_1^{(\pi(1))}\right],\left[\gamma_2^{(\pi(2))}\right],\dots,\left[\gamma_\Lambda^{(\pi(\Lambda))}\right]},$$
    where $\left[\gamma_i^{(j)}\right]$ denotes the $j^\text{th}$ logical in the $i^\text{th}$ copy. 
    This logical operator is at the $\Lambda$-th level in the Clifford hierarchy.
    \item The unitary $U_{\cup, \Lambda}$ can be implemented by the circuit $\gate_{\cup, \Lambda}$ whose depth depends on~$\Lambda$ only.
    \end{enumerate}
      Once $\Lambda$ is fixed, we obtain a family of quantum LDPC codes indexed by $L$ with distance $\op{poly}(L)$, i.e.\ the family has parameters $[[\Lambda^2 L^\Lambda, \Lambda^2, \Theta(L)]]_{L}$.
      
   The circuit depth of the copy-cup circuit $\gate_{\cup, \Lambda}$ is constant with respect to the code size and, thus, the realized gate is trivially fault-tolerant.
\end{theorem}
\begin{proof}
    The explicit description of $U_{\cup, \Lambda}$ as a logical operator follows from \Cref{thm:explicitdescriptionofgate}. The circuit $\gate_{\cup, \Lambda}$ implements $U_{\cup, \Lambda}$ by \Cref{thm:operationyieldsgate}.

    To see that the circuit $\gate_{\cup, \Lambda}$ is constant depth, consider the following graph. Vertices are the qubits and they are connected it they in the same $\gateZ$-gate. 
    The maximum degree of this graph is bounded by some function of $\Lambda$, say $f(\Lambda)$. 
    Choose an $(f(\Lambda)+1)$-coloring of the graph. 
    Note that the number of gates a qubit is involved in is a function of $\Lambda$, say $g(\Lambda)$. 
    Now apply all $\gateZ$-gates involving all qubits of the first color. 
    % Now apply all $\gateZ$ involving qubits of the first color. 
    By construction this can be done by a circuit of depth $g(\Lambda)$. 
    Repeat the process for the remaining colors while excluding gates that have been applied already.
\end{proof}

Using hyperbolic surfaces instead of tori, we obtain a \hl{Poly-rate} family with logarithmic distance as follows.

\begin{theorem}[\hl{Poly-rate} family with gate at level $\Lambda$] \label{th:constantrate}
    Let $\{ H_g \}_g$ be a family of hyperbolic surfaces of genus $g$ with systole growing as $\Theta(\log g)$, each cellulated $\Theta(g)$ with number of cells and incidence between cells bounded by a constant (for example, see Refs.~\cite{macaj2008injectivity,7456305}).

    For $\Lambda$ even, let $X=(H_g)^{\Lambda/2}\uplus \dots \uplus (H_g)^{\Lambda/2}$ be $\Lambda$ copies of the $\Lambda/2$-fold Cartesian product $(H_g)^{\Lambda/2}$.
    By placing qubits on the edges of $X$, we obtain a quantum code with $n=\Theta(\Lambda^2g)$ physical qubits and \hl{$\Theta( n^{\frac{2}{\Lambda}})$} logical qubits and distance $\op{polylog}(n)$. Then:
    \begin{enumerate}
        \item  The copy-cup unitary $U_{\cup, \Lambda}$ on $X$ implements a logical operator on the $\Lambda^\text{th}$ level in the Clifford hierarchy.
        \item This unitary can be implemented with a circuit $\gate_{\cup, \Lambda}$ which has of depth depending only on~$\Lambda$.
    \end{enumerate}
    For a fixed $\Lambda$, we obtain a quantum LDPC code family indexed by $g$ with a constant rate.
    Furthermore, the circuit depth of the copy-cup circuit $\gate_{\cup, \Lambda}$ is constant with respect to the code size.  
\end{theorem}
\begin{proof}
    Analogous to \Cref{thm:codeblockstoricgates}.
\end{proof}
We note that we can also interpolate the parameters of \Cref{thm:codeblockstoricgates,th:constantrate} by fixing some functional dependence between $L$ and $g$ and employing the construction of~\cite{semi-hyp}.
In particular, choosing $L\sim g$ gives a family with polynomially vanishing rate and relative distance.

Lastly, we consider tensor products of classical codes.

\begin{theorem}[Conditional; Poly-distance and constant-rate family with gate at level~$\Lambda$] \label{th:best_cup_gate}
   Let $\{C({}_{\ell}X,\Zn{2})\}_\ell$ be a family of classical codes with constant rate and with distance and dual distance in $\op{poly}(\ell).$ Define the property:
    \begin{itemize}
        \item[$(\star)$] For each $\ell$, $C({}_{\ell}X,\Zn{2})$ is equipped with pre-orientation satisfying the conditions of \cref{prop:integrated-leibniz-conditions}  In addition, assume that the integrated $\Lambda$-fold cup product map is nontrivial on each classical complex.
    \end{itemize}
   
    Consider the quantum code given by $\Lambda$ copies of the $\Lambda$-fold tensor product $C({}_{\ell}X,\Zn{2})^{\otimes \Lambda}$. Conditional on property $(\star)$,
    \begin{enumerate}
        \item The copy-cup unitary $U_{\cup, \Lambda}$ implements a logical operator on the $\Lambda^\text{th}$ level of the Clifford hierarchy on this code.
        \item  This unitary can be implemented with a circuit $\gate_{\cup, \Lambda}$ which has of depth depending only on~$\Lambda$.
    \end{enumerate}
    Once $\Lambda$ is fixed, we obtain a quantum LDPC code family indexed by $\ell$ of constant rate and poly-distance, see \cite{product_code_distance}.
    Furthermore, the circuit depth of the copy-cup circuit $\gate_{\cup, \Lambda}$ is constant with respect to the code size and thus trivially fault-tolerant.
\end{theorem}
\hl{We conjecture that condition $(\star)$ can be satisfied. }
One should be able to also use iterated balanced products of Abelian codes as long as the group action preserves pre-orientation and the integral. For example, one can consider the balanced products of Sipser--Spielman codes.

\subsection{Toy examples}

Let us also provide two intuition-building examples of an explicit circuit realizing $\Lambda = 2$-fold and $\Lambda=3$-fold cup product gates on two copies of 2D toric code and three copies of 3D toric code, respectively. 

\subsubsection{$\op{CZ}$ gate in 2D toric codes on a square lattice}

We provide a detailed example of \Cref{thm:codeblockstoricgates} for $\Lambda=2$.

 Consider the chain complex associated with the repetition code $$C^0(S^1,\Zn{2})\stackrel{\delta}{\to} C^1(S^1,\Zn{2})$$ where the circle $S^1$ is equipped with a simplicial structure and hence with a cup product, making it a $\op{dg}$-algebra.
In \Cref{sec:tensor_products_cup}, we discussed how tensor products of $\op{dg}$-algebras obtain $\op{dg}$-algebra structure via \Cref{eqn:tensor_products_cup}.
These are particularly simple in our case, where the product is a square grid on a torus, see \Cref{ex:tensor_product_torus}.
In particular, we have that two edges have a non-zero cup product if and only if one is ``horizontal'' and the other ``vertical'' and they have consecutive orientation:
\begin{center}
    \begin{tikzpicture}[scale=0.7]
        \begin{scope}[decoration={markings,mark=at position 0.5 with {\arrow{>}}}] 
            \draw[postaction={decorate}] (1,0) -- (3,0);
            \draw[postaction={decorate}] (0,1) -- (0,3);
            % square
            \draw[postaction={decorate}] (1,1) -- (1,3);
            \draw[postaction={decorate}] (1,1) -- (3,1);
            \draw[postaction={decorate}] (1,3) -- (3,3);
            \draw[postaction={decorate}] (3,1) -- (3,3);
        \end{scope}
        \node[scale=1.5] at (0,0) {$\otimes$};
        % gates
        \draw[red,thick] (1,2.3) to[out=0,in=-90] (1.7,3);
        \draw[red,thick] (2.3,1) to[out=90,in=-180] (3,1.7);
    \end{tikzpicture}
\end{center}

The associated copy-cup circuit $\gateZ_{\cup, \Lambda}$ for two copies of the 2D toric code consists of $\op{CZ}$ gates (shown in red below) that are applied between an edge/qubit of one copy (shown in black) to the paired edge/qubit of the other copy (shown in blue).
\begin{center}
    \begin{tikzpicture}
    \def\L{5}
    \pgfmathsetmacro{\Lm}{\L - 1}
    \def\shif{.2}
    \clip (.5,.5) rectangle (\Lm+.5,\Lm+.5);
    % First grid with thicker lines
    \draw[step=1cm, black, line width=0.15mm] (0, 0) grid (\L, \L);
    % gates
    % bottom to right
    \foreach \x in {0,...,\Lm} {
        \foreach \y in {0,...,\Lm} {
            \draw[red, line width = 0.25 mm] ({\x+.8},\y) to[out=90,in=-180] ({\x+1+\shif},{\y+.4}) ;
        }
    }
    % left to top
    \foreach \x in {0,...,\Lm} {
        \foreach \y in {0,...,\Lm} {
            \draw[red,  line width = 0.3 mm] (\x,\y+.8) to[out=0,in=-90] ({\x+.4},{\y+1+\shif}) ;
        }
    }
    % Second grid, slightly offset
    \begin{scope}[xshift=\shif cm, yshift=\shif cm]
        \draw[step=1cm, blue, line width=0.2mm] (0, 0) grid (\L, \L);
    \end{scope}
\end{tikzpicture}
\end{center}
It is easy to check that $X$-checks of either code, which are supported on the coboundaries of vertices, are mapped onto $Z$-checks on the other copy, which are supported on the boundaries of faces.
Similarly, an $X$-logical on either copy, supported on a representative of a cohomology class, is mapped onto a $Z$-logical on the other copy.
We therefore conclude that the circuit indeed implements the logical operator $\overline{\op{CZ}}_{[\gamma_1^{(1)}],[\gamma_2^{(2)}]}\overline{\op{CZ}}_{[\gamma_1^{(2)}],[\gamma_2^{(1)}]}$.
This gate has appeared inexplicitly in terms of pumping of topological defects in~\cite{Yoshida16,Potter17}, and later explicitly in~\cite{Barkeshli23}.

\subsubsection{$\op{CZ}$ in the anisotropic lineon code}
Consider the anisotropic lineon code (also called the 2-foliated fracton code) introduced in \cite{ShirleySlagleChen2019,ShirleySlagleChen2019_2}, which is a geometrically local code in three dimensions. Qubits are placed on the $z$ edge and the $xy$ plaquettes of the cubic lattice. The $X$ and $Z$ stabilizers are shown as follows
\begin{center}
    \includegraphics[]{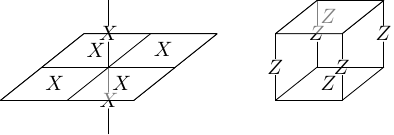}    
\end{center}

The above code can be constructed from a tensor product of the repetition code and a classical code called the plaquette Ising (Xu-Moore) code \cite{XuMoore2004}, constructed from a square lattice with periodic boundary conditions. The bits (black dots) correspond to vertices and each check (white dot) is a plaquette that checks the four surrounding vertices (note that although the code is presented on a 2D ``lattice", it is a classical code corresponding to a 1D chain complex). On a square lattice of size $L\times L$, the plaquette Ising code has parameters $[L^2,2L-1,L]$. The logicals correspond to flipping bits on each of the $L$ rows or the $L$ columns separately. However, flipping all rows and columns simultaneously leaves all bits unflipped, which result in $2L -1$ independent logicals.

The plaquette Ising code can be assigned a pre-orientation by denoting one outgoing/incoming bit (black dot) for each check (white dot) with an outgoing/incoming arrow. Specifically, $\delta _{\op{in}}$ is the bit to the southwest, $\delta _{\op{out}}$ is the bit to the northeast, and $\delta _{\op{free}}$ are the bits to the northwest and southeast.
\begin{center}
    \includegraphics[scale=0.5]{plaquette.pdf}
\end{center}
The pre-orientation satisfies the non-overlapping bits condition \cref{def:non-overlapping} and the integrated Leibniz rule \cref{def:int-Leibniz} for $\Lambda=2$.

We now consider the quantum code which is constructed from the tensor product. Using the pre-orientation, the cup product is non-zero for the following pairs of qubits inside each cube
\begin{center}
    \includegraphics[]{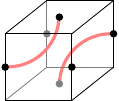}    
\end{center}
The associated copy-cup circuit $\gateZ_{\cup, \Lambda}$ of two copies of this quantum code consists of the following CZ gates applied to each pair of qubits.
\begin{center}
    \includegraphics[]{Fractoncup2.pdf}    
\end{center}
The resulting logical action is non-trivial, and is given by a pairing of $Z$ logicals between the two copies such that the conjugate $\op{X}$-logical in the second copy has a similar support to the that of the $\op{Z}$-logical in the first copy. Thus at the logical level, it acts as a product of $4L-2$ logical $\op{CZ}$'s.

\subsubsection{$\op{CCZ}$ gate in 3D toric codes on a cubic lattice}
Similarly, we may consider the case $\Lambda=3$ by taking a triple tensor product of the repetition code $$\left(C^0(S^1,\Zn{2})\stackrel{\delta}{\to} C^1(S^1,\Zn{2})\right)^{\otimes 3},$$ which gives the 3D toric code on a cubic lattice.
Again, we obtain a $\op{dg}$-algebra structure with a cup product.
In particular, the cup-product of three edges $e\cup e' \cup e''$ is nonzero if and only if they point in each of the $x,y,z$-directions and form an oriented path along the induced orientation (such as the blue and orange paths highlighted below). There is a total of 6 such paths in a cube. 
\begin{center}
    \begin{tikzpicture}[scale=0.6]
        \def\L{5}
        \coordinate (O) at (0,0,0);
        \coordinate (A) at (0,\L,0);
        \coordinate (B) at (0,\L,\L);
        \coordinate (C) at (0,0,\L);
        \coordinate (D) at (\L,0,0);
        \coordinate (E) at (\L,\L,0);
        \coordinate (F) at (\L,\L,\L);
        \coordinate (G) at (\L,0,\L);

        \begin{scope}[decoration={markings,mark=at position 0.5 with {\arrow{>}}}] 
            \draw[postaction={decorate}] (C) -- (O);
            \draw[postaction={decorate}] (C) -- (G);
            \draw[postaction={decorate}] (G) -- (D);
            \draw[postaction={decorate}] (O) -- (A);
            \draw[postaction={decorate}] (A) -- (E);
            \draw[postaction={decorate}] (D) -- (E);
            \draw[postaction={decorate}] (C) -- (B);
            \draw[postaction={decorate}] (F) -- (E);
            \draw[postaction={decorate}] (O) -- (D);
            \draw[postaction={decorate}] (B) -- (A);
            \draw[postaction={decorate}] (B) -- (F);
            \draw[postaction={decorate}] (G) -- (F);
        \end{scope}

        % tensor product
        \node[scale=1.5] at (-1.4,-1,\L+1) {$\otimes$};
        \begin{scope}[decoration={markings,mark=at position 0.5 with {\arrow{>}}}] 
            \draw[postaction={decorate}] (-1.4,-1,\L) -- (-1.4,-1,0);
            \draw[postaction={decorate}] (-.4,-1,\L+1) -- (\L-.4,-1,\L+1);
            \draw[postaction={decorate}] (-1.4,0,\L+1) -- (-1.4,\L,\L+1);
        \end{scope}

        % paths
        \draw[line width = 1.5 mm, orange, opacity=.4] (C) -- (G) -- (D) -- (E);
        \draw[line width = 1.5 mm, blue, opacity=.4] (C) -- (B) -- (F) -- (E);

    \end{tikzpicture}
\end{center}
For each cube, the copy-cup circuit $\gateZ_{\cup, \Lambda}$ applies a $\op{CCZ}$-gate between qubits each copy  along \ each of the six paths. 
We can pick a logical basis for each copy by taking a $xy$-, $yz$- and $xz$-plane and defining the support as all edges having one vertex on that plane and the other vertex in the positive $z$-, $x$- and $y$-direction, respectively.
Now we can observe that these supports only give rise to a non-zero cup product at a single cube that lies at the intersection of all three planes.
The implemented logical operator is therefore
\begin{align*}
    & \overline{\op{CCZ}}_{[\gamma_1^{(1)}],[\gamma_2^{(2)}],[\gamma_3^{(3)}]}\, \overline{\op{CCZ}}_{[\gamma_1^{(1)}],[\gamma_2^{(3)}],[\gamma_3^{(2)}]}\, \overline{\op{CCZ}}_{[\gamma_1^{(2)}],[\gamma_2^{(1)}],[\gamma_3^{(3)}]}\, \\
\times&\overline{\op{CCZ}}_{[\gamma_1^{(2)}],[\gamma_2^{(3)}],[\gamma_3^{(1)}]}\, \overline{\op{CCZ}}_{[\gamma_1^{(3)}],[\gamma_2^{(1)}],[\gamma_3^{(2)}]}\, \overline{\op{CCZ}}_{[\gamma_1^{(3)}],[\gamma_2^{(2)}],[\gamma_3^{(1)}]}
\end{align*}
where the indices refer to the logicals in each copy.
Note that they are in one-to-one correspondence to the paths in the cube.

Similarly to the two-dimensional case, one can check that $X$-checks of any of the codes are mapped onto  $CZ$-type operators between the two other copies that act trivially on the codespace.
An $X$-logical on any of the copies is mapped onto a $CZ$ logical gate between the other two copies.
This circuit appeared implicitly in Refs.~\cite{Barkeshli23,Zhu22fractal,Zhu2023,song2024magic} and explicitly in Refs.~\cite{Chen2023higher,Wang2023}

\section*{Acknowledgments}
We thank Shankar Balasubramanian, Andi Bauer, Benjamin Brown, Yu-An Chen, Andrew Cross, Chris Pattison, Brenden Roberts, Lev Spodyneiko and Ryan Tiew for helpful discussions, as well as an anonymous referee for pointing out a mistake in an earlier version of this paper.
NPB and NT also thank the Centro de Ciencias de Benasque Pedro Pascual, where part of this work was performed.

\section*{Funding and Competing interests}
The work of MD and NT was supported by the Walter Burke Institute for Theoretical Physics at Caltech. 
This work was done in part while NPB and MD were visiting the Simons Institute for the Theory of Computing, supported by DOE QSA grant \#FP00010905.
We declare no conflicts of interest.

\section*{Data Availability}
No datasets were generated or analysed during the current study.

\bibliography{bibliography.bib}
\end{document}